\documentclass[11pt, tightenlines, onecolumn, amsmath, amssymb, notitlepage, a4paper, nofootinbib, superscriptaddress]{revtex4-1}

\usepackage[strict]{revquantum}
\usepackage{graphicx}
\usepackage{subfig}
\usepackage{braket}
\usepackage[margin=1in]{geometry}
\usepackage{bbm}
\usepackage{enumitem}

\DeclareMathOperator{\tr}{tr}
\DeclareMathOperator{\E}{\mathbb{E}}
\DeclareMathOperator{\pr}{Pr}
\newcommand{\norm}[1]{\left\lVert#1\right\rVert}

\newcommand{\nocontentsline}[3]{}
\newcommand{\tocless}[2]{\bgroup\let\addcontentsline=\nocontentsline#1{#2}\egroup}

\newtheorem{corollary}{Corollary}[]

\usepackage{thmtools}
\usepackage{thm-restate}

\begin{document}

\title{\Large Efficient online quantum state estimation using a matrix-exponentiated gradient method}
\author{Akram Youssry}
\affiliation{University of Technology Sydney,
		Centre for Quantum Software and Information,
		Ultimo NSW 2007, Australia}
\affiliation{Department of Electronics and Communication Engineering, Faculty of Engineering, Ain Shams University, Cairo, Egypt} 

\author{Christopher Ferrie}
\affiliation{University of Technology Sydney,
		Centre for Quantum Software and Information,
		Ultimo NSW 2007, Australia}

\author{Marco Tomamichel}
\affiliation{University of Technology Sydney,
		Centre for Quantum Software and Information,
		Ultimo NSW 2007, Australia}

\date{\today}

\begin{abstract}

In this paper, we explore an efficient online algorithm for quantum state estimation based on a matrix-exponentiated gradient method previously used in the context of machine learning. The state update is governed by a learning rate that determines how much weight is given to the new measurement results obtained in each step. We show convergence of the running state estimate in probability to the true state for both noiseless and noisy measurements. We find that in the latter case the learning rate has to be chosen adaptively and decreasing to guarantee convergence beyond the noise threshold. 
As a practical alternative we then propose to use running averages of the measurement statistics and a constant learning rate to overcome the noise problem. The proposed algorithm is numerically compared with batch maximum-likelihood and least-squares estimators. The results show a superior performance of the new algorithm in terms of accuracy and runtime complexity.
\end{abstract}

\maketitle
\tableofcontents

\section{Introduction}
The field of quantum information processing has grown rapidly over the past decade, largely motivated by the wide range of prospective applications of quantum computing, quantum cryptography, and quantum communications. However, building scalable quantum devices is still an enormous challenge. A core unsolved problem is the efficient characterization of quantum systems of intermediate size---can we check efficiently whether a quantum device comprised of a few qubits performs as intended? Practical considerations and, in particular, efficiency of the estimation procedure are at the forefront as quantum systems move beyond the curiosity of experimental physics to prototype quantum technology devices. 

The most fundamental characterization problem concerns state estimation---determining an unknown state of a quantum system using a series of different measurements. This procedure is referred to as quantum state estimation or quantum state tomography. Quantum state estimation usually refers to estimating the state using incomplete information, whereas quantum state tomography is often used to describe the situation where complete (and sometimes even noise-free) information about the state is assumed. The two terms can be used interchangeably, though we stick to the former throughout the paper. The literature on quantum state estimation is extensive (see, e.g. the survey text \cite{QSTbook}) with methods ranging from simple linear inversion to least-squares (LS) regression \cite{qi_quantum_2013}, maximum likelihood (ML) estimation \cite{hradil_quantum-state_1997, rehacek_diluted_2007}, methods based on compressed sensing \cite{gross_quantum_2010, flammia_quantum_2012}, and the Bayesian approach (see, e.g., \cite{granade_practical_2016}). The maximum likelihood method is considered optimal in the sense that it yields a valid state that maximizes the probability of the observed data, and converges to the true state in the limit of many measurements. A disadvantage of the method is that it often yields estimates at the boundary of the state set, i.e.\ states that are rank deficient.

Gradient-based approximation methods \cite{bolduc_projected_2017, shang_superfast_2017}, promise to be much faster but they can produce non-physical states (with the estimate either having negative eigenvalues, or being unnormalized) and convergence is in many cases not guaranteed. The former problem can be solved in practice by projecting the state back into the physical space~\cite{smolin_efficient_2012}. The same problem is also present in linear regression methods. The matrix exponentiated gradient (MEG) method has found use in classical machine learning \cite{tsuda_matrix_2005,globerson_exponentiated_2007} and offers an appealing alternative as it by construction ensures positive semidefiniteness of the matrix estimate. In \cite{li_general_2017}, MEG was applied to perform quantum tomography on qubits and approximate the maximum likelihood estimate efficiently. In this paper, we chose MEG among other online estimation methods as we can show strong convergence results. Other efficient methods such as projected-gradients would be also interesting to explore, but this is outside the scope of this paper.

In this work, we use the MEG technique to devise an efficient online estimator for quantum states. Our algorithm satisfies the following three desiderata: (1) it is online---providing a running estimate of the state as data is collected; (2) it is fast---its runtime scales well with the dimension of the system; and (3) it comes with a convergence proof. Many other techniques satisfy some of these properties, but we are not aware of any that satisfy all. 
The main results of our work can be summarized as follows. 
\begin{itemize}[itemsep=0pt,topsep=5pt]
\item We present the MEG algorithm suitable for online quantum state estimation and robust to noise.
\item We prove convergence for noiseless and noisy measurements.
\item We numerically compare one of the proposed algorithms with online versions of ML and LS estimators and find that it converges equally fast.
\item The proposed algorithm is computationally more efficient than other approaches (such as online versions of ML and LS), scaling as $O(d^3)$ instead of $O(d^4)$, where $d$ is the dimension of the quantum system.
\end{itemize}

Our algorithm is naturally \emph{online}, which makes it interesting for many applications. For example, when large amounts of measurements need to be taken to verify a state or when the state is likely to change over time, it can be beneficial to have a running estimate that allows for a rapid diagnosis of error. While any batch algorithm (like the maximum likelihood estimator) can be run on a subset of the the initial data points to create an online estimate, this creates a significant overhead and can be avoided using an online estimator.

\paragraph*{Related work:}
A different perspective on quantum state learning has been taken in \cite{aaronson2007learnability,aaronson2018shadow} where instead of learning a full description of the state the goal is only to predict future measurement outcomes. Concurrent with our work, this approach has also been generalized to the online setting in \cite{aaronson_online_2018}, also using variations of the MEG method. 
The main difference is that their work targets obtaining predictions of future measurement outcomes based on previous ones, which can be achieved without full state tomography. 
The authors show, somewhat surprisingly, that this can be done up to constant error using only a number of measurements linear in the number of qubits. In contrast full characterization requires exponentially many measurements (see, e.g.~\cite{haah2017sample}).
Second, the error criterion to be minimized is based on a mistake bound (i.e. the number of time steps where the prediction was far from the true value), whereas we aim to show asymptotic convergence to the true state. 
A technical consequence of this is that in \cite{aaronson_online_2018} the learning rate can be chosen to be a constant whereas we find that for convergence a decreasing learning rate is necessary.


\section{Summary of main results}

Let us first describe the MEG update rule (see also Section~\ref{megmethod} for more details). We assume that the true state, $\rho$, is finite-dimensional. The update algorithm takes four inputs: $\hat{\rho}_t$ is the estimate of $\rho$ calculated in the previous step; $X_t$ and $\hat{y}_t$ are the observable and measurement outcome at time step $t$; and $\eta_t$ is the learning rate at time step $t$. The algorithm then returns the next estimate of the state, $\hat{\rho}_{t+1}$, as follows.

\begin{algorithm}[H]
	\caption{Matrix-exponentiated gradient update rule for quantum state estimation}
	\label{alg:MEG}
	\begin{algorithmic}
	\Function{Update}{$\hat{\rho}_t$, $X_t$, $\hat{y}_t$, $\eta_t$}
	\State $G_{t+1} \gets \log(\hat{\rho}_t) - 2\eta_t (\tr(\hat{\rho}_t X_t)-\hat{y}_t)X_t$
	\Comment{correct by the gradient of the loss function}
	\State \Return $\hat{\rho}_{t+1} \gets \frac{\exp(G_{t+1})}{\tr\exp(G_{t+1})}$
    \Comment{our next estimate, properly normalized}
	\EndFunction
	\end{algorithmic}
\end{algorithm}

First, we introduce the use of MEG for online quantum state estimation in the ideal case where there is no noise in the measurements. This case may approximate the situation where experimentally a very large number of shots of each measurement are taken. The number of shots refers to the number of copies of the state that are needed to estimate the counts of each possible outcome. So, first the initial estimate is chosen arbitrarily to be the completely mixed state, i.e.\ $\hat{\rho}_1=\frac{1}{d}I_d$. Next, a measurement operator $X_t$ is selected at random, and the noiseless measurement is done to obtain $\hat{y}_t = \tr(\rho X_t)$. In this setting the learning rate is chosen to be any constant such that $0<\eta<\frac{1}{2}$. Finally, the estimate is updated according to the MEG rule as in Algorithm \ref{alg:MEG}. The estimate in this case converges in probability to the true state if the random set of measurements form a unitary one-design, e.g.\ if they are Pauli measurements in the case of one or more qubits. In other words, we show that for all $\delta > 0$,
\begin{align}
	 \lim_{t\to\infty}{\Pr \left\{\norm{\hat{\rho}_t-\rho}_F <\delta\right\}}=1, \label{eq:conv}
\end{align}
where $\| \cdot \|_F $ denotes the Frobenius norm (or any other matrix norm) and the probability is taken over the choice of measurements. In fact, we can show that convergence in Frobenius norm is essentially as fast as $1/\sqrt{t}$ in the following sense. For any $\alpha \in (0, \frac12)$, we have
\begin{align}
	\lim_{T\to\infty} {\Pr {\left\{\norm{\hat{\rho}_t - \rho}_F < \frac{1}{t^{\alpha}}\right\} } }= 1 \,. \label{eq:conv2}
\end{align} 
Here the probability is taken over the measurement choices as well as over $t$ uniformly chosen from the set $\{1, 2, \ldots, T\}$. Essentially this tells us that the probability of a random $t$ exceeding the bound $1/t^{\alpha}$ vanishes, even though we cannot guarantee that the bound is satisfied for any fixed $t$.
The proof of this behavior is presented in Section \ref{sec:proof_noiseless}.

\begin{figure*}[t!]
	\centering
	\subfloat[1-qubit]{\includegraphics[width=0.48\textwidth]{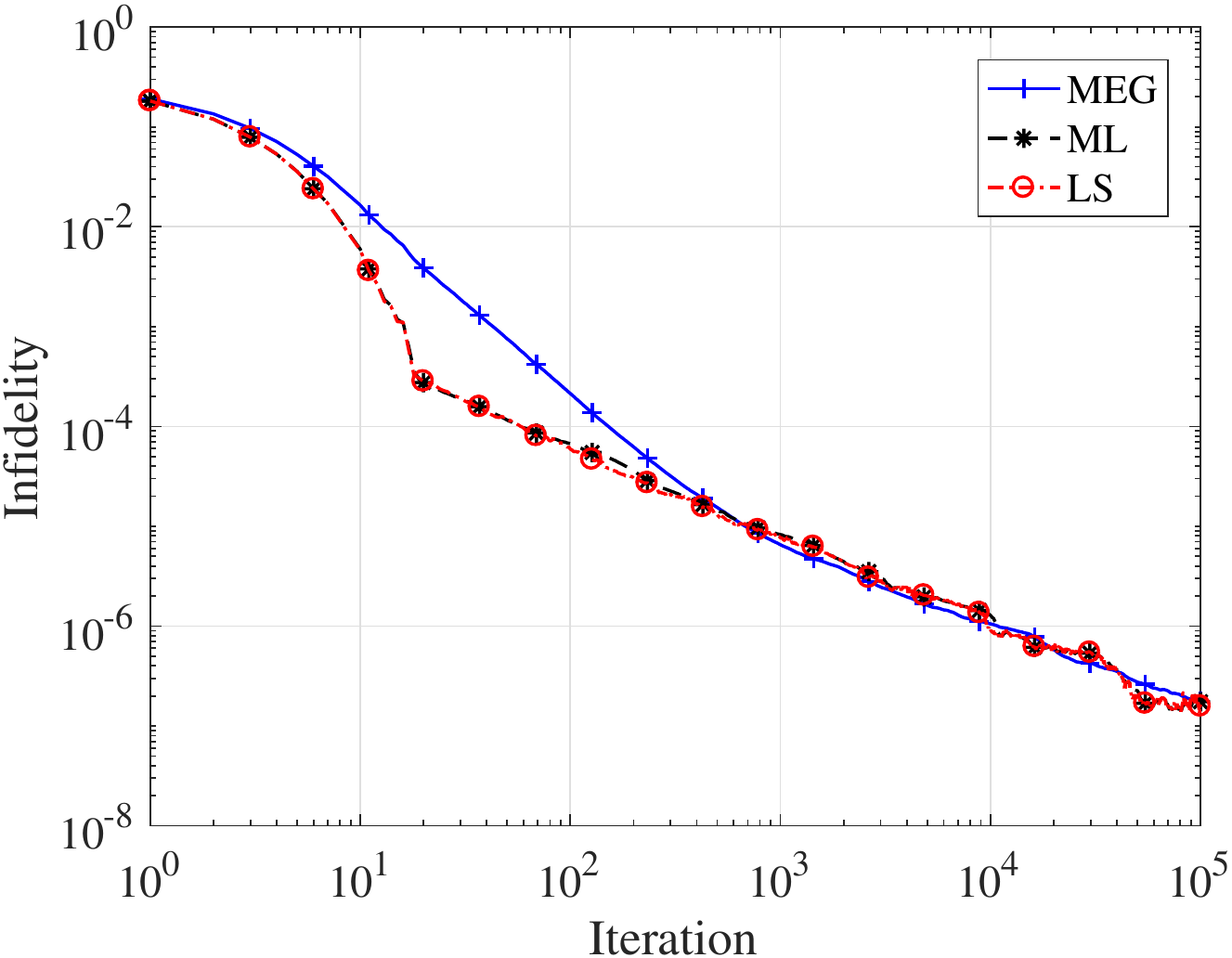}}
	\subfloat[2-qubit]{\includegraphics[width=0.48\textwidth]{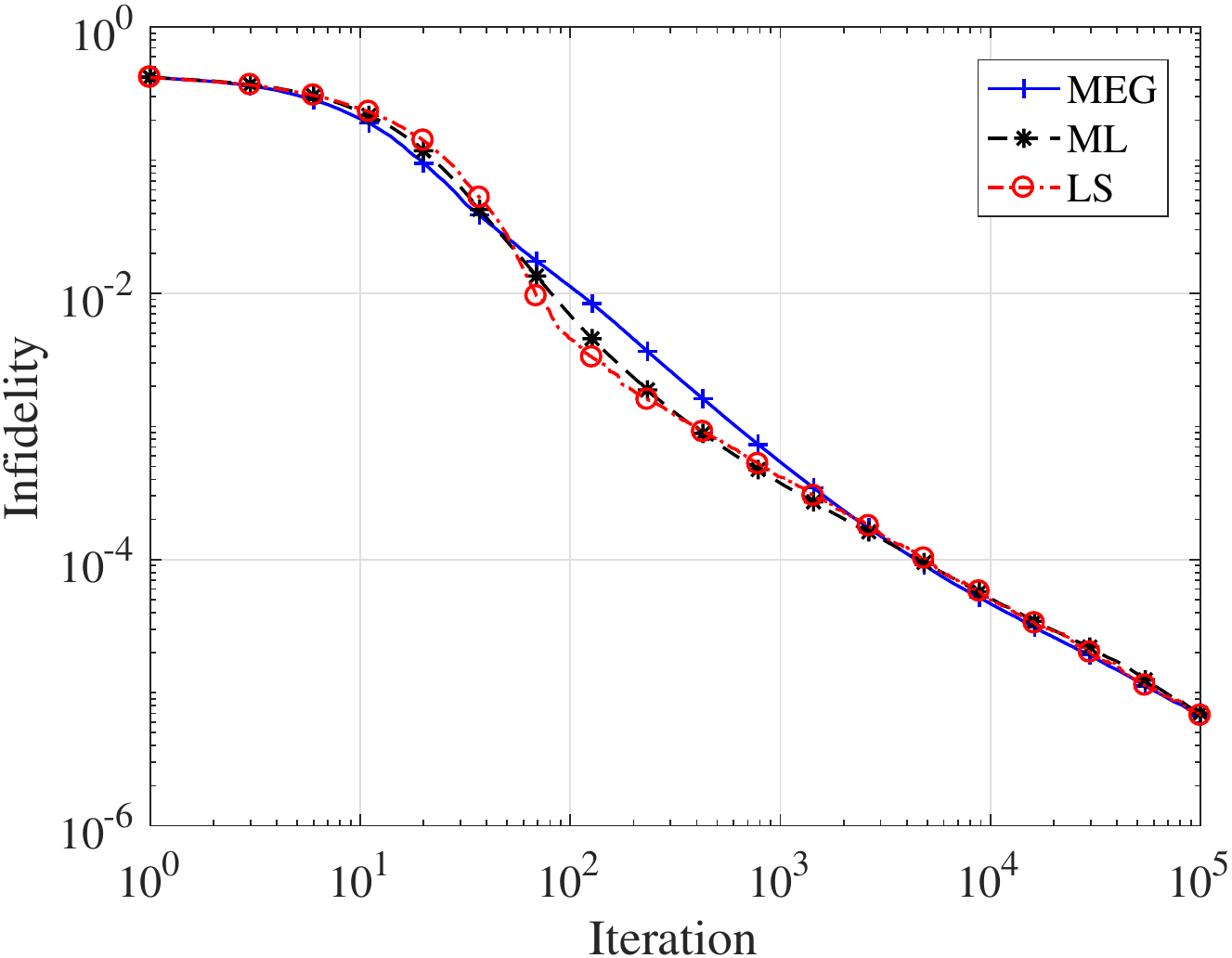}}\\
	\subfloat[3-qubit]{\includegraphics[width=0.48\textwidth]{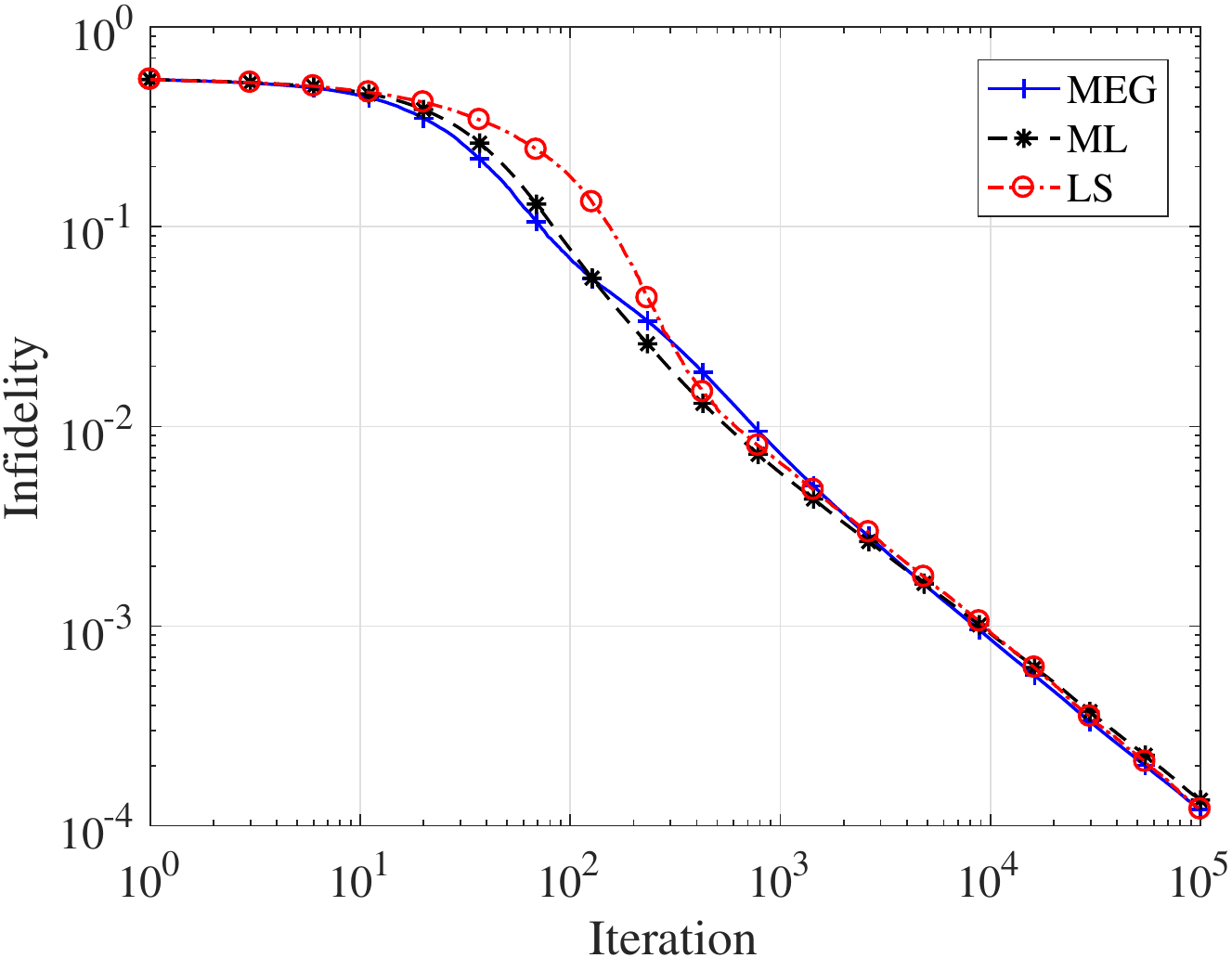}}
	\subfloat[4-qubit]{\includegraphics[width=0.48\textwidth]{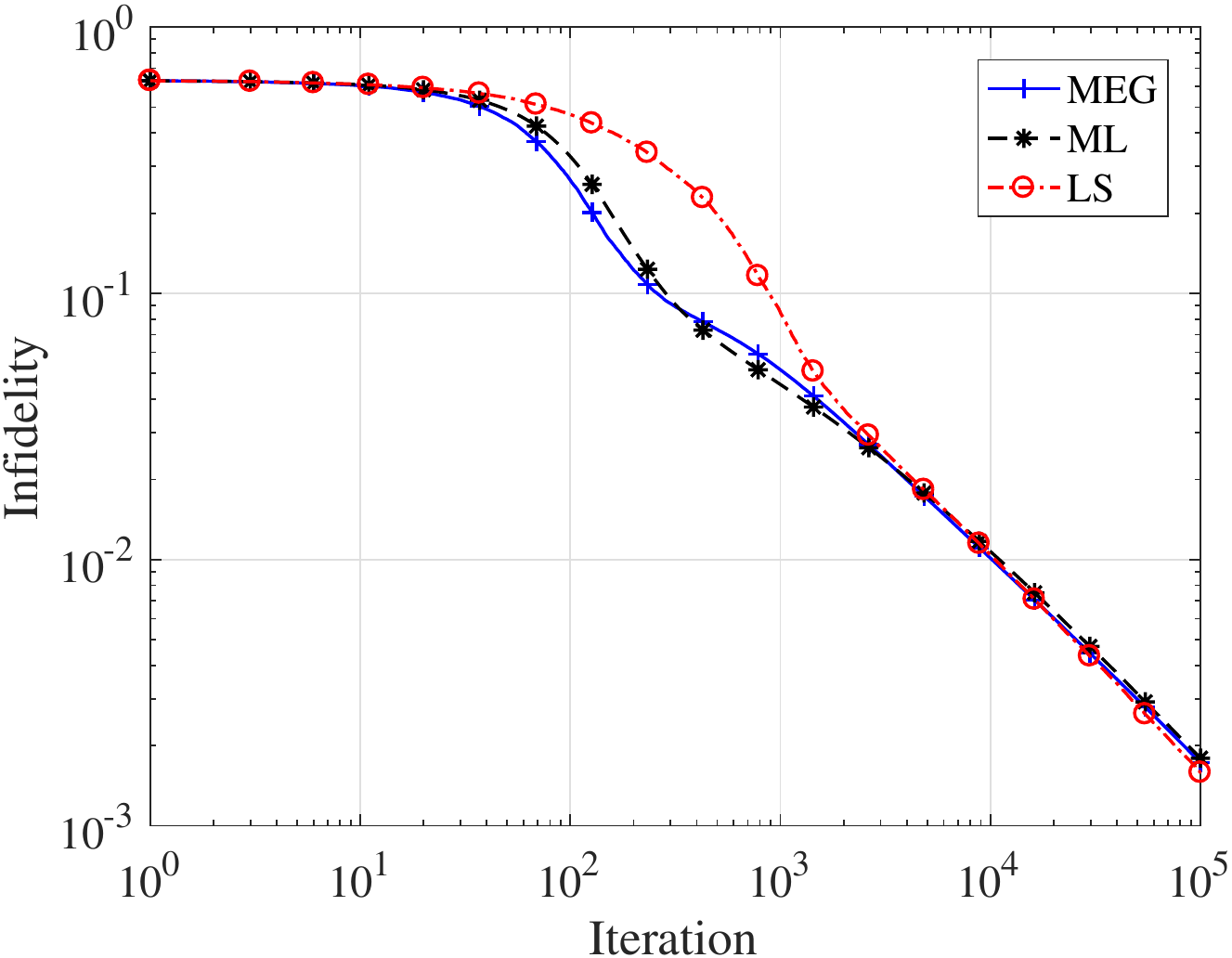}}
	
	\caption[]{Simulation results for different multi-qubit systems: (a) 1-qubit, (b) 2-qubit, (c) 3-qubit, and (d) 4-qubit. The infidelity is averaged over 1000 randomly generated quantum states and plotted versus the iteration number. The three lines correspond to the proposed matrix exponential gradient (MEG) method, maximum likelihood (ML) estimator and least-squares (LS) estimator. The number of shots per measurement is taken to be 1000 shots. }
	\label{fig:Results_multi}
\end{figure*}

Let us next discuss the (more realistic) case of noisy measurements.  Here we are taking a finite number of shots per measurement so that $\hat{y}_t$ is a random variable with mean $\tr(\rho X_t)$ and a variance that depends on the number of shots.
In this noisy case the previous scheme will not converge. To see this, assume at some iteration we hit the true state, $\hat{\rho}_t=\rho$. We then see that even for this state the gradient will be non-zero because in general $\hat{y}_t \ne \tr(\rho X_t)$ and thus the update rule will push the estimate away from the true state. To avoid this behavior, we propose a scheme with an adaptive, decreasing learning rate. We show that a convergence guarantee in the form of~\eqref{eq:conv2} holds, although the convergence will be slower. To achieve this, for any $\alpha \in (0, \frac14)$, we set the learning rate to $\eta_t = \frac14 t^{-\beta}$ with $\beta = \frac34 - \alpha$ to find that the MEG algorithm satisfies
\begin{align}
	\lim_{T\to\infty} {\Pr {\left\{\norm{\hat{\rho}_t - \rho}_F < \frac{1}{t^{\alpha}}\right\} } }= 1 \,,
\end{align} 
where the probability is taken over the measurement choices and outcomes, as well as $t$ uniformly from the set $\{1, 2, \ldots, T\}$.
Section \ref{sec:proof_noisy} discusses the proof of this statement.

Finally, for our numerical testing in low dimensions we propose another approach to solve the problem with noisy measurements by using a running average of the measurement outcomes for each measurement. This is effectively equivalent to increasing the number $N$ of shots when certain measurements are repeated. This means that eventually the algorithm approaches the noise-free case and convergence is thus ensured (we leave this as an informal statement). Moreover, numerical simulations show that this method converges faster than using an adaptive learning rate. Figure~\ref{fig:Results_multi} compares the convergence of our algorithm to an ML and LS estimator for 1-, 2-, 3- and 4-qubit systems, showing that the proposed algorithm converges to the other two methods. We use infidelity between the true state $\rho$ and the estimate $\hat{\rho}_t$ as an accuracy measure, which is defined as $1- \left( \tr\left| \sqrt{\rho\!\!\phantom{\hat{\rho}}}\sqrt{\hat{\rho}_t}\right| \right)^2$. So, in terms of accuracy measured by infidelity, MEG can perform as well as other methods. Further numerical results can be found in Section~\ref{sec:sims}.

In terms of complexity however, MEG outperforms the other methods with complexity of $O(d^3)$ per update compared to $O(d^4)$ for ML and LS. The bottleneck for MEG is the matrix exponentiation step in the update as seen in Algorithm~\ref{alg:MEG}. 

\section{Preliminaries}\label{megmethod}
We give a detailed description of the problem of online quantum state estimation and an overview of the matrix-exponentiated gradient (MEG) update rule.

\subsection{Problem statement}

Given a quantum system in an unknown state $\rho$, it is required to find an estimated quantum state $\hat{\rho}$, based on the classical outcomes of some measurements performed on copies of the system. The system has dimensions $d$, and so for the case of an $m$-qubit system, we have $d=2^m$. For the numerical simulations in this paper we consider such $m$-qubit systems and perform Pauli measurements on each individual qubit. We shall denote the set of measurements operator by $\{X^{[i]}\}_{i=1}^{d^2-1}$. 

The outcome of such a binary measurement is a classical bit. We shall call these outcomes ``up'' and ``down'' corresponding to the $\pm 1$ eigenvalues of the Pauli operator. In order to do tomography, we assume that we have an ensemble of identically prepared quantum systems in the same unknown state $\rho$, so we can perform independent measurements on each of the subsystems, and calculate the average outcome. So, selecting a measurement operator $X_t = X^{[i(t)]}$ at time step $t$, the expected value of the measurement denoted by $y_t$ as predicted by the Born rule is given by $y_t = \tr(\rho X_t)$,
while the actual average we calculate if we repeat the experiment $N$ times is the random variable 
\begin{align}
	\hat{y}_t=\frac{n_\uparrow-n_\downarrow}{N}=\frac{2n_\uparrow - N}{N}.
\end{align}
Here, $n_\uparrow$ is the number of times the ``up'' outcome was observed, while $n_\downarrow$ is the number of times the ``down'' outcome was observed. We know that $n_\uparrow$ follows a binomial distribution. Given a measurement operator represented in terms of its eigenvalue projectors as $X_t=\Pi_\uparrow-\Pi_\downarrow$, we have 
$n_\uparrow \sim B(N,p)$ with $p=\tr(\rho\Pi_\uparrow)$. It is then easy to verify that
\begin{align}
	\E\{\hat{y}_t\} = 2p-1 = y_t\,, \quad \textrm{and} \quad 
	\text{Var}\{\hat{y}_t\} = \frac{4 p(1-p)}{N} = \frac{1-y_t^2}{N} \,. \label{equ:var}
\end{align}

We can then repeat the whole procedure and obtain a sequence of data points in the form $\{ (X_1,\hat{y}_1), ...(X_t,\hat{y}_t),...\}$. Notice that the measurement outcomes $\hat{y}_t$ form an independent and identically distributed (i.i.d.) set of random variables.
Since we are proposing an online algorithm, we do not have the whole data set in advance. We obtain one point at a time, and use it to update an estimate $\hat{\rho}_t$ of the true state. We would like that $\hat{\rho}_t$ converges to $\rho$ as $t$ increases.


\subsection{The matrix-exponentiated gradient (MEG) method}

The MEG method was proposed in \cite{tsuda_matrix_2005,globerson_exponentiated_2007} for some classical machine learning applications and symmetric matrices. The algorithm trivially generalizes to Hermitian matrices. Given a new data point $(X_t,\hat{y}_t)$, the loss function at time step $t$, evaluated for a general quantum state $\sigma$, is defined as
\begin{align}
	L_t(\sigma):=(\tr(\sigma X_t)-\hat{y}_t)^2.
\end{align}
The gradient of the loss function at time step $t$ is then
\begin{align}
	\nabla L_t=2(\tr(\sigma X_t)-\hat{y}_t)X_t.
\end{align}

Consider now the following online cost function
\begin{align}
	D(\hat{\rho}_{t+1}||\hat{\rho}_t) + \eta_t L_t(\hat{\rho}_{t+1}),
\end{align}
where $D$ is Umegaki's quantum relative entropy~\cite{umegaki62} defined as $D(\rho||\sigma) = \tr(\rho\log(\rho)-\rho\log(\sigma))$ for any two states $\rho$ and $\sigma$, and $\eta_t$ is the learning rate. This cost function represents two conflicting goals. The first one is to have an estimate that is near the previous estimate, quantified by the relative entropy. This is important because in the online setting of the problem, we do not want the algorithm to forget what it has learnt so far. The second goal is to move the new estimate so that the loss function at the new data point is hopefully smaller. The learning rate $\eta_t$ controls this trade-off. Minimizing the cost function with respect to $\hat{\rho}_{t+1}$ by taking the gradient (see Appendix A in \cite{tsuda_matrix_2005} for the details of the calculation) and setting it to zero results in
\begin{align}
	\log(\hat{\rho}_{t+1}) = \log(\hat{\rho}_{t}) - \eta\nabla{L_t}(\hat{\rho}_{t+1}) -I,
\end{align}
where $I$ denotes the identity matrix.
Now, since we cannot find an explicit form for $\hat{\rho}_{t+1}$, we may approximate $\hat{\rho}_{t+1}$ by $\hat{\rho}_t$ in the gradient to arrive at
$\log(\hat{\rho}_{t+1}) = \log(\hat{\rho}_{t}) - \eta\nabla{L_t}(\hat{\rho}_{t}) - I $, or, equivalently,
\begin{align}
	\hat{\rho}_{t+1}=\exp \left( \log\left(\hat{\rho}_t\right)-\eta\nabla L_t(\hat{\rho}_t) - I \right).
\end{align}
This form of the update rule ensures that if we start with a positive definite matrix $\hat{\rho}_t$, and a Hermitian operator $X_t$, then we are sure that the new estimate $\hat{\rho}_{t+1}$ is positive definite. This is because the terms inside the exponential function are Hermitian, and thus the matrix exponential results in a positive definite matrix.  Next, we want to make sure that the estimate has unit trace, to be a valid quantum state. So, we normalize to finally obtain the MEG rule:
\begin{align}
	\hat{\rho}_{t+1}=\frac{\exp \left( \log\left(\hat{\rho}_t\right)-\eta\nabla L_t (\hat{\rho}_t)\right)}{\tr\left(\exp\left( \log\left(\hat{\rho}_t\right)-\eta\nabla L_t(\hat{\rho}_t)\right)\right)} .
\end{align}

The update rule can also be expressed in the following compact alternative form:
\begin{align}
	G_{t}&=G_{t-1} - \eta \nabla L_t(\hat{\rho}_t),\quad G_0=\log(\hat{\rho}_0),\quad \textrm{and}\quad \hat{\rho}_{t} =\frac{\exp(G_{t})}{\tr\exp(G_{t})} \,.
\end{align}

\section{Convergence analysis}
\label{sec:proofs}
This section starts with stating some bounds related to the MEG update rule. Next, the proof of convergence for the noise-free case is given, followed by the proof of convergence in the noisy case. Finally, a discussion about the proposed running-average technique is presented. Some additional proofs are provided in Appendix \ref{appendix:proofs}.
\subsection{General bounds on the loss functions for the MEG rule}
We will start by stating the following lemma which bounds the normalization constant that appears in the MEG update rule 
\begin{align}
	\log(\hat{\rho}_{t+1})=\log(\hat{\rho}_t)+\delta_t X_t- \log(\tr(\exp(\log(\hat{\rho}_t)+\delta_t X_t))),
\end{align}
where
\begin{align}
	\delta_t=-2\eta(\tr(\hat{\rho}_t X_t)-\hat{y}_t),
\end{align}
and measurement operators satisfying $-I \le X_t \le I$ to ensure that the updated estimate has unit trace. This bound will be used to prove other important results. The proof is given in Appendix \ref{apndx:lem:log_Zt} generalizing the methods that involved real symmetric matrices in \cite{tsuda_matrix_2005} to complex Hermitian matrices.

\begin{restatable}{lemma}{logzt}
\label{lem:log_Zt}
The normalization constant in the MEG rule update is bounded by
\begin{align}
	\log(\tr(\exp(\log(\hat{\rho}_t)+\delta_t X_t))) \le \frac{\delta_t^2 }{2} +\delta_t\tr\left(\hat{\rho}_t X_t\right).
\end{align}
\end{restatable}
Next, we state the following lemma which puts a bound on the difference between the loss function evaluated at the estimate, and a general state. The lemma relates this difference to the progress of the estimator towards that general state. This is the main lemma that will be used to prove the convergence of MEG. Appendix \ref{apndx:lem:Lt_inequ} gives the proof generalizing the results \cite{tsuda_matrix_2005} to the quantum setting.
\begin{restatable}{lemma}{ltinequ}
\label{lem:Lt_inequ}
Given the loss function $L_t(\hat{\rho}_t)= (\tr(\hat{\rho}_t X_t)-\hat{y}_t)^2$ with measurement operators $-I \le X_t \le I$ and learning rate $0<\eta<\frac{1}{2}$, then for any state $\sigma$,
\begin{align}
	\eta L_t(\hat{\rho}_t) - \frac{\eta}{1-2\eta} L_t(\sigma) \le D(\sigma||\hat{\rho}_t)-D(\sigma||\hat{\rho}_{t+1}).
\end{align}
\end{restatable}
This leads to the following corollary that bounds the loss function of the estimate when the true state is used as the comparison state, in the case of noise-free measurements (i.e.\ $\hat{y}_t=y_t$).  
\begin{corollary}
\label{cor:loss_rho}
Given the loss function $L_t(\hat{\rho}_t)= (\tr(\hat{\rho}_t X_t)-y_t)^2$ with measurement operators $-I \le X_t \le I$ and learning rate $0<\eta<\frac{1}{2}$. Then, given the true state $\rho$, the following relation holds:
\begin{align}
	\eta L_t(\hat{\rho}_t) \le D(\rho||\hat{\rho}_t)-D(\rho||\hat{\rho}_{t+1}).
\end{align}

\begin{proof}
Apply Lemma \ref{lem:Lt_inequ}, set $\sigma =\rho$, and use the fact that $L_t(\rho)=0$.
\end{proof}

\end{corollary}

\subsection{Convergence analysis for noise-free measurements}
\label{sec:proof_noiseless}

The choice of measurements for doing quantum state estimation is arbitrary. However, in this paper we consider the case of performing local Pauli measurements on each qubit of a multi-qubit system. This facilitates the experimental realization compared to performing some other, possibly global, measurement. The proofs will start by calculating some expectation values  involving Pauli operators and loss functions. These results will be used to prove the main theorem showing the convergence of MEG in the noise-free case. We start with the following lemma about the the set of Pauli operators for multi-qubit systems.
\begin{lemma}
\label{lem:U_is_t} 
The set $U = \{U_i\}_{i=0}^{d^2-1}$ of Pauli operators including the identity operator in a $d$-dimensional quantum system satisfy
\begin{align}
	\frac{1}{d}\sum_i{U_i\otimes U_i^{\dagger}}=P_{21},
\end{align}
where $P_{21}$ is the swap operator defined as 
\begin{align}
	P_{21} = \sum_{i,j}{| i \rangle\!\langle j | \otimes | j \rangle\!\langle i |}. 
\end{align}

\begin{proof}
The Pauli's form a unitary orthonormal basis of Hermitian $d \times d$ matrices. Therefore, they form a quantum 1-design due to Proposition 6 in \cite{adam2013applications}. In other words, 
\begin{align}
	\int_{U}{U\rho U^{\dagger} dU}=\sum_i{\frac{1}{d^2} U_i\rho U_i^{\dagger}}.
\end{align} 
Now, from (3.27) and (3.29) in \cite{adam2013applications}, $\sum_i{\frac{1}{d^2} U_i\otimes U_i^{\dagger}}=\frac{P_{21}}{d}$.
\end{proof}

\end{lemma}
Next, we calculate the expectation value of a Pauli operator that is tensored with itself. This calculation will be needed in the calculation of the expectation of the loss function.
\begin{lemma}
\label{lem:EXt} 
The expectation value of the Pauli operators chosen uniformly at random from the set $ U - \{I\}$ satisfies the relation:
\begin{align}
	\E_{X}\left\{X\otimes X\right\}=\frac{d}{d^2-1}P_{21} - \frac{1}{d^2-1}I_d\otimes I_d
\end{align}

\begin{proof}
We have
\begin{align}
	\E_{X}\left\{X\otimes X\right\}&=\frac{1}{d^2-1}\sum_{j=1}^{d^2-1}{X^{[j]}\otimes X^{[j]}}\\
	&=\frac{1}{d^2-1}\left(\sum_{i=0}^{d^2-1}{U_i \otimes U_i}-I_d\otimes I_d\right)\\
	&=\frac{d}{d^2-1}P_{21} - \frac{1}{d^2-1}I_d\otimes I_d,
\end{align}
where the last equality holds from Lemma \ref{lem:U_is_t}, and $I_d$ is the identity operator of dimension $d$.
\end{proof}

\end{lemma}
The following lemma is a commonly-used result in quantum information. The proof is direct---see for example Lemma 1.2.1 in \cite{low2010pseudo}.
\begin{lemma}[Swap trick]
\label{lem:swap} 
For any quantum system with arbitrary dimensions, and for two operators $M$ and $N$, we have
$\tr(M N)=\tr\big((M\otimes N)P_{21}\big)$,
where $P_{21}$ is the swap operator on the quantum system (interchanges any two copies).
\end{lemma}
We are now ready to prove the following lemma in which the expectation of the loss function is calculated.
\begin{lemma}
\label{lem:E_Lt}
Assuming we select the measurement operator $X_t$ at each time iteration uniformly at random from the set $ U - \{I\}$ then for any true state $\rho$ and any state $\sigma$ independent of $X_t$,

\begin{align}
	\E_{X_t}\{L_t(\sigma)\}=\frac{d}{d^2-1}\norm{\sigma-\rho}_F^2.
\end{align}

\begin{proof}
From the definition of the loss function,
\begin{align}
	L_t(\sigma)=(\tr(\sigma X_t)-\tr(\rho X_t))^2.
\end{align}
Taking the expectation of the loss function with respect to $X_t$ we get:
\begin{align}
	\E_{X_t}\{L_t(\sigma)\} &=\E_{X_t}\{(\tr(\sigma X_t)-\tr(\rho X_t))^2\}\\ 
	&= \E_{X_t}\{(\tr(\sigma - \rho)X_t)^2\}\\
	&=\E_{X_t}\left\{\tr\big(\left((\sigma - \rho)X_t\right)\otimes\left((\sigma - \rho)X_t\right)\big)\right\}\\
	&=\E_{X_t}\left\{\tr\big(\left((\sigma - \rho)\otimes (\sigma -\rho)\right) \left(X_t\otimes X_t\right)\big)\right\}\\
	&=\tr\big(\left((\sigma - \rho)\otimes (\sigma - \rho)\right) \E_{X_t}\left\{X_t\otimes X_t\right\}\big).
\end{align}
Then, applying Lemma \ref{lem:EXt}, 
\begin{align}
	\E_{X_t}\{L_t(\sigma\} &=\tr\left(\left((\sigma - \rho)\otimes (\sigma - \rho)\right) \left(\frac{d}{d^2-1}P_{21} - \frac{1}{d^2-1}I\otimes I\right) \right)\\
	&=  \frac{d}{d^2-1}\tr\big(\left((\sigma - \rho)\otimes (\sigma - \rho)\right) P_{21}\big) - \frac{1}{d^2-1}\tr\big(\left((\sigma - \rho)\otimes (\sigma - \rho)\right)(I\otimes I)\big).
\end{align}
Now, applying the swap trick in Lemma \ref{lem:swap},
\begin{align}
	\E_{X_t}\{L_t(\sigma)\} &=\frac{d}{d^2-1}\tr\big((\sigma - \rho)(\sigma - \rho)\big) - \frac{1}{d^2-1}\tr\left(\sigma - \rho\right)\tr\left(\sigma - \rho\right)\\
	&=\frac{d}{d^2-1}\tr\big((\sigma - \rho)^2\big)\\
	&=\frac{d}{d^2-1}\norm{\sigma - \rho}_F^2.
\end{align}
In particular, it is clear that at any time step $t$, if $\sigma \ne \rho$, then $\E\{L_t(\sigma)\}>0$.
\end{proof}

\end{lemma}
Now, we can show the following theorem considering convergence of the noiseless MEG.

\begin{theorem}
\label{lem:weak_conv_noiseless}
The state estimate using the MEG update rule converges in probability to the true state, i.e.\ for any $\delta>0$, 
\begin{align}
	\lim_{t\to\infty}{\Pr\left\{\norm{\hat{\rho}_t-\rho}_F^2<\delta\right\}}=1.
\end{align}
\end{theorem}

\begin{proof}

We know from Corollary \ref{cor:loss_rho} that,
\begin{align}
	\eta L_t(\hat{\rho}_t) \le D(\rho||\hat{\rho}_t)-D(\rho||\hat{\rho}_{t+1}).
\end{align}
Taking the expectation with respect to $X_t$,
\begin{align}
	\eta \E_{X_t}\{L_t(\hat{\rho}_t)\} \le \E_{X_t}\{D(\rho||\hat{\rho}_t)\} - \E_{X_t}\{D(\rho||\hat{\rho}_{t+1})\}.
\end{align}
Applying Lemma \ref{lem:E_Lt}, and using the fact that $\hat{\rho}_t$ is independent of $X_t$ we get
\begin{align}
	\frac{\eta d}{d^2-1}\norm{\hat{\rho}_t-\rho}_F^2 \le D(\rho||\hat{\rho}_t) - \E_{X_t}\{D(\rho||\hat{\rho}_{t+1})\}.
\end{align}
Taking the expectation of the above inequality over all past time iterations $\E$ = $\E_{X_1}...\E_{X_{t-1}}$ we get
\begin{align}
	\frac{\eta d}{d^2-1}\E\{\norm{\hat{\rho}_t-\rho}_F^2\} \le \E\{D(\rho||\hat{\rho}_t)\} - \E\{D(\rho||\hat{\rho}_{t+1})\}.
\end{align}
Next, we sum the inequality over the time iterations to get
\begin{align}
	\frac{\eta d}{d^2-1}\sum_{t=1}^{T}{\E\{\norm{\hat{\rho}_t-\rho}_F^2\}} &\le \E\{D(\rho||\hat{\rho}_1)\} - \E\{D(\rho||\hat{\rho}_{T+1})\} \\ 
	&\le \E\{D(\rho||\hat{\rho}_1)\}.
\end{align}
If we now take the limit as $T\to\infty$, we obtain 
\begin{align}
	\sum_{t=1}^{\infty}{\E\{\norm{\hat{\rho}_t-\rho}_F^2\}} &\le \frac{d^2-1}{\eta d}\E\{D(\rho||\hat{\rho}_1)\}\\
	&= \frac{d^2-1}{\eta d}D(\rho||\hat{\rho}_1).
\end{align}
Where the last line follows from the fact that the true state $\rho$ and the initial estimate $\hat{\rho}_1$ are independent of $X_t$ and $y_t$. Now the right hand side of the inequality is constant, so the series on the left hand side of the inequality converges. This implies by the divergence test that
\begin{align}
	\lim_{t\to\infty}{\E\{\norm{\hat{\rho}_t-\rho}_F^2\}}=0.
\end{align}
Now we can apply Lemma \ref{lem:prob} on the random variable $Z_t = \norm{\hat{\rho}_t-\rho}_F^2$  to conclude that 
\begin{align}
	\forall \delta>0: \quad \lim_{t\to\infty}{\Pr\left\{\norm{\hat{\rho}_t-\rho}_F^2\le\delta\right\}}=1.
\end{align}
Therefore, the estimate $\hat{\rho}_t$ converges in probability to the true state $\rho$.
\end{proof}
Finally, we prove the main theorem that shows a stronger statement for the convergence of MEG algorithm in the case of noise-free measurements.  

\begin{theorem}
Let $\delta \in (0,1)$. Then for any $\alpha \in \left(0,1\right)$, and learning rate $0<\eta<\frac{1}{2}$, there exists $T_0$ given by
\begin{equation}
	T_0 =\left(\frac{\frac{d^2-1}{\eta d} \log{d} +2}{\delta}\right)^\frac{3}{1-\alpha},
\end{equation}
such that for any $T> T_0$ we have,
\begin{align}
	\Pr {\left\{\norm{\hat{\rho}_t-\rho}_F^2 < \frac{1}{t^{\alpha}}\right\} } \ge 1-\delta,
\end{align} 
where the probability is taken over all measurement choices and $t$ uniformly in $\{1, 2, \ldots, T\}$. Moreover, 
\begin{align}
	\lim_{T\to\infty} {\Pr {\left\{\norm{\hat{\rho}_t-\rho}_F^2 < \frac{1}{t^{\alpha}}\right\} } }= 1 \,.
\end{align} 
\begin{proof}

Let the initial estimate be $\hat{\rho}_1 = \frac{I_d}{d}$. We know from Corollary \ref{cor:loss_rho} that,
\begin{align}
	\eta L_t(\hat{\rho}_t) \le D(\rho||\hat{\rho}_t)-D(\rho||\hat{\rho}_{t+1}).
\end{align}
Taking the expectation with respect to $X_t$,
\begin{align}
	\eta \E_{X_t}\{L_t(\hat{\rho}_t)\} \le \E_{X_t}\{D(\rho||\hat{\rho}_t)\} - \E_{X_t}\{D(\rho||\hat{\rho}_{t+1})\}.
\end{align}
Applying Lemma \ref{lem:E_Lt}, and using the fact that $\hat{\rho}_t$ is independent of $X_t$ we get
\begin{align}
	\frac{\eta d}{d^2-1}\norm{\hat{\rho}_t-\rho}_F^2 \le D(\rho||\hat{\rho}_t) - \E_{X_t}\{D(\rho||\hat{\rho}_{t+1})\}.
\end{align}
Taking the expectation of the above inequality over all past time iterations $\E$ = $\E_{X_0}\E_{X_1}...\E_{X_{t-1}}$ we get
\begin{align}
	\frac{\eta d}{d^2-1}\E\{\norm{\hat{\rho}_t-\rho}_F^2\} \le \E\{D(\rho||\hat{\rho}_t)\} - \E\{D(\rho||\hat{\rho}_{t+1})\}.
\end{align}
Next, we sum the inequality over the time iterations to get
\begin{align}
	\frac{\eta d}{d^2-1}\sum_{t=1}^{T}{\E\{\norm{\hat{\rho}_t-\rho}_F^2\}} &\le \E\{D(\rho||\hat{\rho}_1)\} - \E\{D(\rho||\hat{\rho}_{T+1})\}\\ 
	&\le \E\{D(\rho||\hat{\rho}_1)\}\\
	&\le\log{d}.
\end{align}
Now, let $\epsilon_t = \E\{\norm{\hat{\rho}_t-\rho}_F^2\}$, $\delta_t = \frac{1}{t^{\alpha+\gamma}}$, and $\gamma = \frac{2}{3}(1-\alpha)$. Notice that $\alpha+\gamma <1$.
Define the set
\begin{align}
	\mathcal{T} :=\left\{t\in\{1,2,..T\}: \epsilon_t \ge \frac{1}{t^{\alpha+\gamma}} \right\}.
\end{align}
Rearranging the terms in the inequality we get
\begin{align}
	 \frac{d^2-1}{\eta d}\log{d} &\ge \sum_{t=1}^{T}{\epsilon_t}\\
	 &\ge \sum_{t=1}^{T}{\left(\mathbbm{1}_{\epsilon_t >\delta_t}\right)\epsilon_t}\\
	 &\ge \sum_{t=1}^{T}{\left(\mathbbm{1}_{\epsilon_t > \delta_t}\right)\frac{1}{t^{\alpha+\gamma}}}\\
	 &\ge \sum_{t=1}^{T}{\left(\mathbbm{1}_{\epsilon_t >\delta_t}\right)\frac{1}{T^{\alpha+\gamma}}}\\
	 &\ge \frac{|\mathcal{T}_{\delta}|}{T^{\alpha+\gamma}}.
\end{align}
In other words, the ratio between the number of iterations in which $\E\{\norm{\hat{\rho}_t-\rho}_F^2\} \ge \frac{1}{t^{\alpha+\gamma}}$ and the total number of iterations $T$ we performed so far is bounded by 
\begin{align}
	\frac{|\mathcal{T}_{\delta}|}{T} \le K T^{\alpha+\gamma-1},
\end{align}
where $K:=\frac{d^2-1}{\eta d}\log{d}$. This implies that 
\begin{align}
	\lim_{T\to\infty}{\left(\frac{T_{\delta}}{T}\right)}=0,
\end{align}
because $\alpha+\gamma<1$. This means that increasing the number of iterations results in decreasing the number of times where the estimate was not accurate enough. Let's state this formally. Assume we do a total number of iterations $T$. If we select at random a fixed time step $1\le \tilde{t} \le T$, then there will be two possible outcomes. Either $\epsilon_{\tilde{t}} \le \delta_{\tilde{t}}$ or $\epsilon_{\tilde{t}}> \delta_{\tilde{t}}$. Assume we get the first outcome, then by applying Markov's inequality,
\begin{align}
	\epsilon_{\tilde{t}} \le\delta_{\tilde{t}} \implies \Pr\left\{ \norm{\hat{\rho}_{\tilde{t}}-\rho}_F^2\} \ge \frac{1}{\tilde{t}^{\alpha}} \right\} &\le \E\left\{\norm{\hat{\rho}_{\tilde{t}}-\rho}_F^2\right\} \tilde{t}^{\alpha}\\
	&\le \delta_{\tilde{t}}\tilde{t}^{\alpha}\\
	&=\tilde{t}^{-\gamma}.
\end{align}
Now, we can find the joint probability
\begin{align}
	\Pr_{t, \hat{\rho}_t}\left\{ \norm{\hat{\rho}_t-\rho}_F^2 \ge \frac{1}{t^{\alpha}}\right\} &= \Pr\left\{\norm{\hat{\rho}_{t}-\rho}_F^2\} \ge \frac{1}{t^{\alpha}} \middle|t=\tilde{t} \right\}\Pr \left\{ \tilde{t} \in \mathcal{T}_{\delta} \right\} \nonumber \\ &+ \Pr \left\{ \norm{\hat{\rho}_{t}-\rho}_F^2\} \ge \frac{1}{t^{\alpha}} \middle|t=\tilde{t}\right\} \Pr \left\{\tilde{t} \not \in \mathcal{T}_{\delta} \right\} \\
	&\le \Pr \left\{ \tilde{t} \in \mathcal{T}_{\delta} \right\} + \sum_{\tilde{t}\not \in \mathcal{T}_{\delta}}{\Pr \left\{ \norm{\hat{\rho}_{t}-\rho}_F^2\} \ge \frac{1}{t^{\alpha}} \middle|t=\tilde{t}\right\}\frac{1}{T}} \\
	&\le \frac{|\mathcal{T}_{\delta}|}{T} + \frac{1}{T} \sum_{\tilde{t}=1}^{T}{\Pr \left\{ \norm{\hat{\rho}_{t}-\rho}_F^2\} \ge \frac{1}{t^{\alpha}} \middle|t=\tilde{t}\right\}} \\ 
	&\le \frac{|\mathcal{T}_{\delta}|}{T} + \frac{1}{T} \sum_{\tilde{t}=1}^{T}{\tilde{t}^{-\gamma}} \\ 
	&\le \frac{|\mathcal{T}_{\delta}|}{T} + \frac{1}{T} \sum_{\tilde{t}=1}^{T}{\left(\frac{1}{\tilde{t}^2}\right)^{\frac{\gamma}{2}}}. 
\end{align}

Applying Jensen's inequality on the second term (noting that $f(x)=x^r$ is a concave function for $0<r<1$) yields
\begin{align}
	\Pr \left\{ \norm{\hat{\rho}_t-\rho}_F^2 \ge \frac{1}{t^{\alpha}}\right\} &=\frac{|\mathcal{T}_{\delta}|}{T} + \left(\frac{1}{T} \sum_{\tilde{t}=1}^{T}{\frac{1}{\tilde{t}^2}}\right)^{\frac{\gamma}{2}} \\ 
	&\le \frac{|\mathcal{T}_{\delta}|}{T} + \left(\frac{1}{T} \sum_{\tilde{t}=1}^{\infty}{\frac{1}{\tilde{t}^2}}\right)^{\frac{\gamma}{2}}\\
	&=\frac{|\mathcal{T}_{\delta}|}{T} + \left(\frac{\pi^2}{6}\frac{1}{T}\right)^{\frac{\gamma}{2}}\\
	&\le \frac{|\mathcal{T}_{\delta}|}{T} + 2 T^{-\frac{\gamma}{2}}
\end{align}
Therefore,
\begin{align}
	\Pr \left\{ \norm{\hat{\rho}_t-\rho}_F^2 < \frac{1}{t^{\alpha}}\right\} &\ge 1 - \frac{|\mathcal{T}_{\delta}|}{T} - 2T^{-\frac{\gamma}{2}} \\
	&\ge  1 - K T^{\alpha + \gamma -1} - 2T^{-\frac{\gamma}{2}}\\
	&= 1 -  K T^{-\frac{1-\alpha}{3}} - 2T^{-\frac{1-\alpha}{3}}\\
	&= 1 -  T^{-\frac{1-\alpha}{3}} \left(K + 2\right).
\end{align}
Now, let 
\begin{equation}
	T_0 = \left(\frac{K+2}{\delta}\right)^{\frac{3}{1-\alpha}},
\end{equation}
then, if choose $T > T_0$, then
\begin{equation}
	\delta \ge T^{-\frac{1-\alpha}{3}} \left(K + 2\right),
\end{equation}	
or,
\begin{align}
	1 - \delta &\le 1- T^{-\frac{1-\alpha}{3}} \left(K + 2\right)\\
	&\le \Pr \left\{ \norm{\hat{\rho}_t-\rho}_F^2 < \frac{1}{t^{\alpha}}\right\}.
\end{align}	

\end{proof}
\end{theorem}
Notice, that taking the limit as $T\to\infty$ we obtain that $\delta=0$, and therefore
\begin{align}
	\lim_{T\to\infty} {\Pr {\left\{\norm{\hat{\rho}_t-\rho}_F^2 \ge \frac{1}{t^{\alpha}}\right\} } }= 0 \,.
\end{align} 

\subsection{Convergence analysis for noisy measurements}
\label{sec:proof_noisy}
In this part, we show that using an adaptive learning rate with noisy measurements results in the convergence of the MEG estimate to the true state. First, some expectation values will be calculated based on similar techniques discussed in the noiseless case. After that, the optimal adaptive learning rate is derived in such a way to ensure the convergence of the estimate to the true state in probability. However, the learning rate in this case depends on the true state which is not practical. So, finally we show that we can choose a learning rate independent of the true state and prove even a stronger statement of convergence. 


We will start with the following lemma to calculate the expectation value of the noise term that appears in the loss function due to performing finite number of measurements. 
\begin{lemma}
\label{lem:EVar} 
The expectation value of the Pauli operators chosen uniformly at random from the set $ U - \{I\}$ satisfy the relation:

\begin{align}
	\E_{X_t}\left\{\frac{1-y_t^2}{N}\right\}=\frac{d}{N(d^2-1)}\left(d-\norm{\rho}_F^2\right).
\end{align}
\end{lemma}

\begin{proof}
We have
\begin{align}
	\E_{X_t}\left\{\frac{1-y_t^2}{N}\right\}&=\frac{1-\E_{X_t}\{y_t^2\}}{N}\\
	&=\frac{1-\E_{X_t}\{\tr(\rho X_t)^2\}}{N}\\
	&=\frac{1-\E_{X_t}\{\tr((\rho\otimes\rho) (X_t\otimes X_t) )\}}{N}\\
	&= \frac{1-\tr((\rho\otimes\rho) \E_{X_t}\{X_t\otimes X_t \}) }{N}.
\end{align}
Applying now Lemma \ref{lem:EXt}, we get
\begin{align}
	\E_{X_t}\left\{\frac{1-y_t^2}{N}\right\}&=\frac{1-\tr\left((\rho\otimes\rho) (\frac{d}{d^2-1}P_{21} - \frac{1}{d^2-1}I\otimes I) \right)}{N}\\
	&=\frac{1}{N(d^2-1)}\left(d^2-1-d\tr(\rho^2)+\tr(\rho)^2\right)\\
	&=\frac{d}{N(d^2-1)}\left(d-\norm{\rho}_F^2\right),
\end{align}
where the swap trick in Lemma \ref{lem:swap} is used in the second line.
\end{proof}

Next, we give the following lemma to calculate the expectation of the loss function for the case of noisy measurements.
\begin{lemma}
\label{lem:E_Lt_noisy}
Assuming we select the measurement operator $X_t$ at each time iteration uniformly at random from the set $ U - \{I\}$ then for any true state $\rho$ and any state $\sigma$ independent of $X_t$ and $\hat{y}_t$ for any $t$,
\begin{align}
	\E_t\{L_t(\sigma)\}=\frac{d}{d^2-1}\left(\norm{\sigma - \rho}_F^2 + \frac{d-\norm{\rho}_F^2}{N}\right).
\end{align}

\begin{proof}

Recall the noisy loss function,
\begin{align}
	L_t(\sigma) = (\tr(\sigma X_t)-\hat{y}_t)^2.
\end{align}
Note that $\sigma$ is independent of $\hat{y}_t$, but can depend on the previous history. So, the expectation can be calculated as
\begin{align}
	\E_{\hat{y}_t}\{L_t(\sigma)\}&=\E_{\hat{y}_t}\{(\tr(\sigma X_t)-\hat{y}_t)^2\}\\
	&=\tr(\sigma X_t)^2-2\tr(\sigma X_t)\E_t\{\hat{y}_t\}+\E_t\{\hat{y}_t^2\}\\
	&=\tr(\sigma X_t)^2-2\tr(\sigma X_t) y_t + y_t^2+ \frac{1-y_t^2}{N}\\
	&= \left(\tr(\sigma X_t) -  y_t\right)^2 + \frac{1-y_t^2}{N}.
\end{align}
Now, Let's take the expectation with respect to $X_t$ as
\begin{align}
	\E_t\{L_t(\sigma)\}&=\E_{X_t}\E_{\hat{y}_t}\{L_t(\hat{\rho}_t)\}\\
	&=\E_{X_t}\left\{\left(\tr(\sigma X_t) -  y_t\right)^2 + \frac{1-y_t^2}{N}\right\}\\
	&=\E_{X_t}\{\left(\tr(\sigma X_t) -  y_t\right)^2 \} + \E_{X_t}\left\{\frac{1-y_t^2}{N}\right\}\\
	&= \frac{d}{d^2-1}\left(\norm{\sigma - \rho}_F^2 + \frac{d-\norm{\rho}_F^2}{N}\right),
\end{align}
where we used the results of Lemmas \ref{lem:E_Lt} and \ref{lem:EVar} in the last step. Notice that as $N\to\infty$, the result of the noiseless case is recovered.
\end{proof}

\end{lemma}
Consequently, the following result shows that the true state is the optimal state that minimizes the loss function.
\begin{corollary}
The state $\rho$ is the unique state that minimizes the expectation of the noisy loss function, where
\begin{align}
	\E_t\{L_t(\rho)\}=\frac{d}{d^2-1} \frac{d-\norm{\rho}_F^2}{N}.
\end{align}

\end{corollary}
The following theorem shows how to select an adaptive learning rate that results in convergence of the MEG estimate in probability for noisy measurements. The proof is given in Appendix 
\ref{appndx:lem:weak_conv_noisy}.

\begin{restatable}{theorem}{weakconv}
\label{lem:weak_conv_noisy}
In the presence of noise, the state estimate using the MEG update rule with learning rate
\begin{align}
	\eta_t = \frac{1}{2}\frac{\E\{\norm{\hat{\rho}_t-\rho}_F^2\}}{\E\{\norm{\hat{\rho}_t-\rho}_F^2\}+2\left(\frac{d^2-1}{Nd}\right)},
\end{align} 
converges in probability to the true state, i.e.\ for all $\delta > 0$, 
\begin{align}
	\lim_{t\to\infty}{\Pr\left\{\norm{\hat{\rho}_t-\rho}_F^2\le\delta\right\}}=1.
\end{align}
\end{restatable}
The problem with this choice of learning rate, is that it depends on the true state. This might be useful in other applications like state tracking, but it will not be practical for tomography applications, where the true state is unknown. So, we show next that in fact we can select another form of the learning rate that is independent of the true state and show a stronger statement of convergence.
\begin{theorem}
Let $\delta \in (0,1)$, $\alpha \in \left(0,\frac{1}{2}\right)$ and $\beta \in \left(\frac{1}{2},1-\alpha\right)$. If we choose a learning rate of the form
\begin{align}
	\eta_t = \frac{\eta_0}{t^{\beta}} \quad \textrm{with} \quad  \eta_0 < \frac{1}{2}, 
\end{align}
then there exists $T_0$ given by
\begin{equation}
	T_0 =\left(\frac{\frac{d^2-1}{\eta_0 d} \left(\log{d} +\frac{2}{N}\frac{\eta_0^2}{1-2\eta_0}\zeta(2\beta)\right)+2}{\delta}\right)^\frac{3}{1-\alpha-\beta},
\end{equation}
such that for any $T> T_0$ we have,
\begin{align}
	\Pr {\left\{\norm{\hat{\rho}_t-\rho}_F^2 < \frac{1}{t^{\alpha}}\right\} } \ge 1-\delta,
\end{align} 
where the probability is taken over all measurement choices and $t$ uniformly in $\{1, 2, \ldots, T\}$. Moreover, 
\begin{align}
	\lim_{T\to\infty} {\Pr {\left\{\norm{\hat{\rho}_t-\rho}_F^2 < \frac{1}{t^{\alpha}}\right\} } }= 1 \,.
\end{align} 

\begin{proof}
Let the initial estimate be $\hat{\rho}_1 = \frac{I_d}{d}$. We know from Lemma \ref{lem:Lt_inequ} that,
\begin{align}
	\eta_t L_t(\hat{\rho}_t) - \frac{\eta_t}{1-2\eta_t} L_t(\rho) \le D(\rho||\hat{\rho}_t)-D(\rho||\hat{\rho}_{t+1}).
\end{align}
Taking the expectation with respect to $y_t$ followed by the the expectation with respect to $X_t$ we get,
\begin{align}
	\eta_t\E_t\{L_t(\hat{\rho}_t)\} -\frac{\eta_t}{1-2\eta_t}\E_t\{L_t(\rho)\} \le D(\rho||\hat{\rho}_t) - \E_t\{D(\rho||\hat{\rho}_{t+1})\}.
\end{align}
Applying Lemma \ref{lem:E_Lt_noisy}, we get
\begin{align}
	\eta_t\frac{d}{d^2-1}\left(\norm{\hat{\rho}_t-\rho}_F^2 + \frac{d-\norm{\rho}_F^2}{N}\right) - \frac{\eta_t}{1-2\eta_t} \frac{d}{d^2-1}\left(\frac{d-\norm{\rho}_F^2}{N}\right) \le D(\rho||\hat{\rho}_t) - \E_t\{D(\rho||\hat{\rho}_{t+1})\}.
\end{align}
Simplifying this expression and taking the expectation with respect to all previous time instants we get
\begin{align}
	\eta_t\E\{\norm{\hat{\rho}_t-\rho}_F^2\}  - \frac{2\eta_t^2}{1-2\eta_t}\left(\frac{d-\norm{\rho}_F^2}{N}\right) \le \frac{d^2-1}{d}\E\{D(\rho||\hat{\rho}_t) - D(\rho||\hat{\rho}_{t+1})\}.
\end{align}
The second term on the left hand side is a function of the purity of the true state (defined as $\norm{\rho}_F^2$). This term comes from the variance of the noise which varies according to the location of the state. It can be bounded to become
\begin{align}
	\eta_t\E\{\norm{\hat{\rho}_t-\rho}_F^2\}  - \frac{2\eta_t^2}{1-2\eta_t}\left(\frac{d^2-1}{Nd}\right) \le \frac{d^2-1}{d}\E\{D(\rho||\hat{\rho}_t) - D(\rho||\hat{\rho}_{t+1})\}.
\end{align}
Summing up the inequality over different time steps we get 
\begin{align}
	\sum_{t=1}^{T}{\eta_t\E\{\norm{\hat{\rho}_t-\rho}_F^2\}  - \frac{2\eta_t^2}{1-2\eta_t}\left(\frac{d^2-1}{Nd}\right)} &\le \frac{d^2-1}{d}\E\{D(\rho||\hat{\rho}_1) - D(\rho||\hat{\rho}_{T+1})\} \\
	&\le \frac{d^2-1}{d} \E\{D(\rho||\hat{\rho}_1)\}\\
	&\le \frac{d^2-1}{d} D(\rho||\hat{\rho}_1)\\
	&\le\frac{d^2-1}{d}\log{d}.
\end{align}
Now, by choosing learning rate in the form
\begin{align}
	\eta_t = \frac{\eta_0}{t^{\beta}}: \eta_0<\frac{1}{2},
\end{align}
the inequality becomes
\begin{align}
	\frac{d^2-1}{d} \log{d} &\ge \sum_{t=1}^{T}{\frac{\eta_0}{t^{\beta}}\E\{\norm{\hat{\rho}_t-\rho}_F^2\}  - \frac{2\eta_0^2}{t^{2{\beta}} -2\eta_0 t^{\beta}}\left(\frac{d^2-1}{Nd}\right)}\\
	&\ge \sum_{t=1}^{T}{\frac{\eta_0}{t^{\beta}}\E\{\norm{\hat{\rho}_t-\rho}_F^2\}  - \frac{2\eta_0^2}{t^{2\beta}-2\eta_0 t^{2\beta}}\left(\frac{d^2-1}{Nd}\right)}\\
	&\ge  \sum_{t=1}^{T}{\frac{\eta_0}{t^{\beta}}\E\{\norm{\hat{\rho}_t-\rho}_F^2\}} - \frac{2\eta_0^2}{1-2\eta_0}\left(\frac{d^2-1}{Nd}\right)\sum_{t=1}^{T}{\frac{1}{t^{2\beta}}}\\
	&\ge \sum_{t=1}^{T}{\frac{\eta_0}{t^{\beta}}\E\{\norm{\hat{\rho}_t-\rho}_F^2\}} - \frac{2\eta_0^2}{1-2\eta_0} \left(\frac{d^2-1}{Nd}\right)\sum_{t=1}^{\infty}\frac{1}{t^{2\beta}}\\			
	&= -\frac{2\eta_0^2}{1-2\eta_0}\left(\frac{d^2-1}{Nd}\right)\zeta(2\beta) + \sum_{t=1}^{T}{\frac{\eta_0}{t^{\beta}}\E\{\norm{\hat{\rho}_t-\rho}_F^2\}}\\	
	&\ge -\frac{2\eta_0^2}{1-2\eta_0}\left(\frac{d^2-1}{Nd}\right)\zeta(2\beta)+ \sum_{t=1}^{T}{\frac{\eta_0}{T^{\beta}}\E\{\norm{\hat{\rho}_t-\rho}_F^2\}},
\end{align}
where $\zeta(\cdot)$ is the Riemann zeta function. Now, let $\epsilon_t = \E\{\norm{\hat{\rho}_t-\rho}_F^2\}$, $\delta_t= \frac{1}{t^{\alpha+\gamma}}$, and $\gamma = \frac{2}{3}(1-\alpha-\beta)$. Notice that $\alpha+\beta+\gamma <1$ as long as $\alpha+\beta<1$. Define the set
\begin{align}
	\mathcal{T}_{\delta}:=\left\{t\in\{1,2,..T\}: \epsilon_t \ge \delta_t\right\}.
\end{align}
Rearranging the terms in the inequality we get
\begin{align}
	 \frac{d^2-1}{\eta_0 d} \left(\log{d} +\frac{2}{N}\frac{\eta_0^2}{1-2\eta_0}\zeta(2\beta)\right) &\ge \frac{1}{T^{\beta}}\sum_{t=1}^{T}{\epsilon_t}\\
	 &\ge \frac{1}{T^{\beta}}\sum_{t=1}^{T}{\left(\mathbbm{1}_{\epsilon_t >\delta_t}\right)\epsilon_t}\\
	 &\ge \frac{1}{T^{\beta}}  \sum_{t=1}^{T}{\left(\mathbbm{1}_{\epsilon_t > \delta_t}\right)\frac{1}{t^{\alpha+\gamma}}}\\
	 &\ge \frac{1}{T^{\beta}}  \sum_{t=1}^{T}{\left(\mathbbm{1}_{\epsilon_t >\delta_t}\right)\frac{1}{T^{\alpha+\gamma}}}\\
	 &\ge \frac{|\mathcal{T}_{\delta}|}{T^{\alpha+\beta+\gamma}}.
\end{align}
In other words, the ratio between the number of iterations in which $\E\{\norm{\hat{\rho}_t-\rho}_F^2\} \ge \frac{1}{t^{\alpha+\gamma}}$ and the total number of iterations $T$ we performed so far is bounded by 
\begin{align}
	\frac{|\mathcal{T}_{\delta}|}{T} \le K T^{\alpha+\beta+\gamma-1},
\end{align}
where $K:=\frac{d^2-1}{\eta_0 d} \left(\log{d} +\frac{2}{N}\frac{\eta_0^2}{1-2\eta_0}\zeta(2\beta)\right)$. This implies that 
\begin{align}
	\lim_{T\to\infty}{\left(\frac{T_{\delta}}{T}\right)}=0,
\end{align}
because $\alpha+\beta+\gamma<1$. This means that increasing the number of iterations results in decreasing the number of times where the estimate was not accurate enough. Let's state this formally. Assuming we do a total number of iterations $T$, then if we select at random a fixed time step $1\le \tilde{t} \le T$, then there will be two possible outcomes. Either $\epsilon_{\tilde{t}} \le \delta_{\tilde{t}}$ or $\epsilon_{\tilde{t}}> \delta_{\tilde{t}}$. Assume we get the first outcome, then by applying Markov's inequality,
\begin{align}
	\epsilon_{\tilde{t}} \le\delta_{\tilde{t}} \implies \Pr\left\{ \norm{\hat{\rho}_{\tilde{t}}-\rho}_F^2\} \ge \frac{1}{\tilde{t}^{\alpha}} \right\} &\le \E\left\{\norm{\hat{\rho}_{\tilde{t}}-\rho}_F^2\right\} \tilde{t}^{\alpha}\\
	&\le \delta_{\tilde{t}}\tilde{t}^{\alpha}\\
	&=\tilde{t}^{-\gamma}.
\end{align}
Now, we can find the joint probability
\begin{align}
	\Pr_{t, \hat{\rho}_t}\left\{ \norm{\hat{\rho}_t-\rho}_F^2 \ge \frac{1}{t^{\alpha}}\right\} &= \Pr\left\{\norm{\hat{\rho}_{t}-\rho}_F^2\} \ge \frac{1}{t^{\alpha}} \middle|t=\tilde{t} \right\}\Pr \left\{ \tilde{t} \in \mathcal{T}_{\delta} \right\} \nonumber \\ &+ \Pr \left\{ \norm{\hat{\rho}_{t}-\rho}_F^2\} \ge \frac{1}{t^{\alpha}} \middle|t=\tilde{t}\right\} \Pr \left\{\tilde{t} \not \in \mathcal{T}_{\delta} \right\} \\
	&\le \Pr \left\{ \tilde{t} \in \mathcal{T}_{\delta} \right\} + \sum_{\tilde{t}\not \in \mathcal{T}_{\delta}}{\Pr \left\{ \norm{\hat{\rho}_{t}-\rho}_F^2\} \ge \frac{1}{t^{\alpha}} \middle|t=\tilde{t}\right\}\frac{1}{T}} \\
	&\le \frac{|\mathcal{T}_{\delta}|}{T} + \frac{1}{T} \sum_{\tilde{t}=1}^{T}{\Pr \left\{ \norm{\hat{\rho}_{t}-\rho}_F^2\} \ge \frac{1}{t^{\alpha}} \middle|t=\tilde{t}\right\}} \\ 
	&\le \frac{|\mathcal{T}_{\delta}|}{T} + \frac{1}{T} \sum_{\tilde{t}=1}^{T}{\tilde{t}^{-\gamma}} \\ 
	&\le \frac{|\mathcal{T}_{\delta}|}{T} + \frac{1}{T} \sum_{\tilde{t}=1}^{T}{\left(\frac{1}{\tilde{t}^2}\right)^{\frac{\gamma}{2}}}. 
\end{align}
Applying Jensen's inequality on the second term (noting that $f(x)=x^r$ is a concave function for $0<r<1$). Thus, 
\begin{align}
	\Pr_{t, \hat{\rho}_t}\left\{ \norm{\hat{\rho}_t-\rho}_F^2 \ge \frac{1}{t^{\alpha}}\right\} &=\frac{|\mathcal{T}_{\delta}|}{T} + \left(\frac{1}{T} \sum_{\tilde{t}=1}^{T}{\frac{1}{\tilde{t}^2}}\right)^{\frac{\gamma}{2}} \\ 
	&\le \frac{|\mathcal{T}_{\delta}|}{T} + \left(\frac{1}{T} \sum_{\tilde{t}=1}^{\infty}{\frac{1}{\tilde{t}^2}}\right)^{\frac{\gamma}{2}}\\
	&=\frac{|\mathcal{T}_{\delta}|}{T} + \left(\frac{\pi^2}{6}\frac{1}{T}\right)^{\frac{\gamma}{2}}\\
	&\le \frac{|\mathcal{T}_{\delta}|}{T} + 2T^{-\frac{\gamma}{2}}
\end{align}
Therefore,
\begin{align}
	\Pr \left\{ \norm{\hat{\rho}_t-\rho}_F^2 < \frac{1}{t^{\alpha}}\right\} &\ge 1 - \frac{|\mathcal{T}_{\delta}|}{T} - 2T^{-\frac{\gamma}{2}} \\
	&\ge  1 - K T^{\alpha + \beta + \gamma -1} - 2T^{-\frac{\gamma}{2}}\\
	&= 1 -  K T^{-\frac{1-\alpha-\beta}{3}} - 2T^{-\frac{1-\alpha-\beta}{3}}\\
	&= 1 -  T^{-\frac{1-\alpha-\beta}{3}} \left(K + 2\right).
\end{align}
Now, let 
\begin{equation}
	T_0 = \left(\frac{K+2}{\delta}\right)^{\frac{3}{1-\alpha-\beta}},
\end{equation}
then, if choose $T > T_0$, then
\begin{equation}
	\delta \ge T^{-\frac{1-\alpha-\beta}{3}} \left(K + 2\right),
\end{equation}	
or,
\begin{align}
	1 - \delta &\le 1- T^{-\frac{1-\alpha-\beta}{3}} \left(K + 2\right)\\
	&\le \Pr \left\{ \norm{\hat{\rho}_t-\rho}_F^2 < \frac{1}{t^{\alpha}}\right\}.
\end{align}	
Now, taking the limit as $T\to\infty$ we obtain finally that,
\begin{align}
	\lim_{T\to\infty} {\Pr_{t, \hat{\rho}_t} {\left\{\norm{\hat{\rho}_t-\rho}_F^2 \ge \frac{1}{t^{\alpha}}\right\} } }= 0,
\end{align} 
or, equivalently,
\begin{align}
	\lim_{T\to\infty} {\Pr_{t, \hat{\rho}_t} {\left\{\norm{\hat{\rho}_t-\rho}_F^2 < \frac{1}{t^{\alpha}}\right\} } }= 1.
\end{align} 
\end{proof}
\end{theorem}

\subsection{Convergence of MEG in the noisy case with averaging}
As discussed previously, doing the running-average over the measurements with a small number of shots is equivalent to increasing the number of shots without having to do this experimentally per each measurement. This method also does not require the use of an adaptive learning rate. So, given the data point $(X_t,\hat{y}_t)$, we calculate the running average $\bar{y}_t$: 

\begin{align}
	\bar{y}_t = \frac{\hat{y}_{r_1}+\hat{y}_{r_2}+...\hat{y}_{r_{n-1}}+\hat{y}_{t}}{n_{X_t}}=\frac{(n_{X_t}-1)\bar{y}_{t-1}+\hat{y}_t}{n_{X_t}},
\end{align}
such that $\{r_i\}_{i=1}^{n_{X_t}} = \{t^\prime:X_{t^\prime}=X_t\}$ are the time indices in which the measurement operator $X_t$ appeared before (which means that $r_n = t$), and $n_{X_t}$ is the number of times it appeared until time $t$. If we are choosing the measurement operators randomly then after enough number of iterations we may assume that we visited all operators the same number of iterations. So as $t\to\infty, n_{X_t}\to\infty$. Now from the strong law of large numbers:
\begin{align}
	\frac{1}{n_{X_t}}\sum_{i=1}^{n_{X_t}}{\hat{y}_{r_j}} \to \E\{\hat{y}_t\} = y_t\quad a.s.
\end{align}
So the gradient of the loss function satisfies that: 
\begin{align}
	\nabla L_t(\hat{\rho}_t)=2(\tr(\hat{\rho}_t X_t)-\bar{y}_t)X_t \quad \to \quad  2(\tr(\hat{\rho}_t X_t)-\tr(\rho X_t))X_t \quad a.s.
\end{align} 
In other words, after enough number of iterations, the situation becomes similar to the noise-free measurements case which allows the possibility of convergence to the true state with a constant learning rate. 

\section{Simulation results \label{sec:sims}}

This section discusses the methods and results of the numerical simulations. An overview of the simulations settings is given first, followed by discussion on the significance of the results. 

\subsection{Methods}
In order to assess the performance of the proposed method, we created a dataset consisting of 1000 randomly generated quantum states for 1-, 2-, 3-, 4-, and 5-qubit systems, as well as simulating 100000 random measurement outcomes for 10, 100, 1000, 10000 shots for each of these states. The estimate after each measurement is calculated, and compared to the true state using the infidelity measure defined as
\begin{align}
  1-F(\rho,\hat{\rho}_t)=1- \left( \tr\left| \sqrt{\rho\!\!\phantom{\hat{\rho}}}\sqrt{\hat{\rho}_t}\right| \right)^2.
\end{align}
Figure \ref{fig:eta} shows the behavior of MEG under different learning rates in the form $\eta_t = 0.5 t^{-\beta}$, compared to using the running average (RA) method with a constant learning rate. The plot shows that using the running average leads to the fastest convergence compared to the case of adaptive learning rate. So, we choose the RA method for further discussion.

\begin{figure}[t]
\centering
\includegraphics[width=0.55\textwidth]{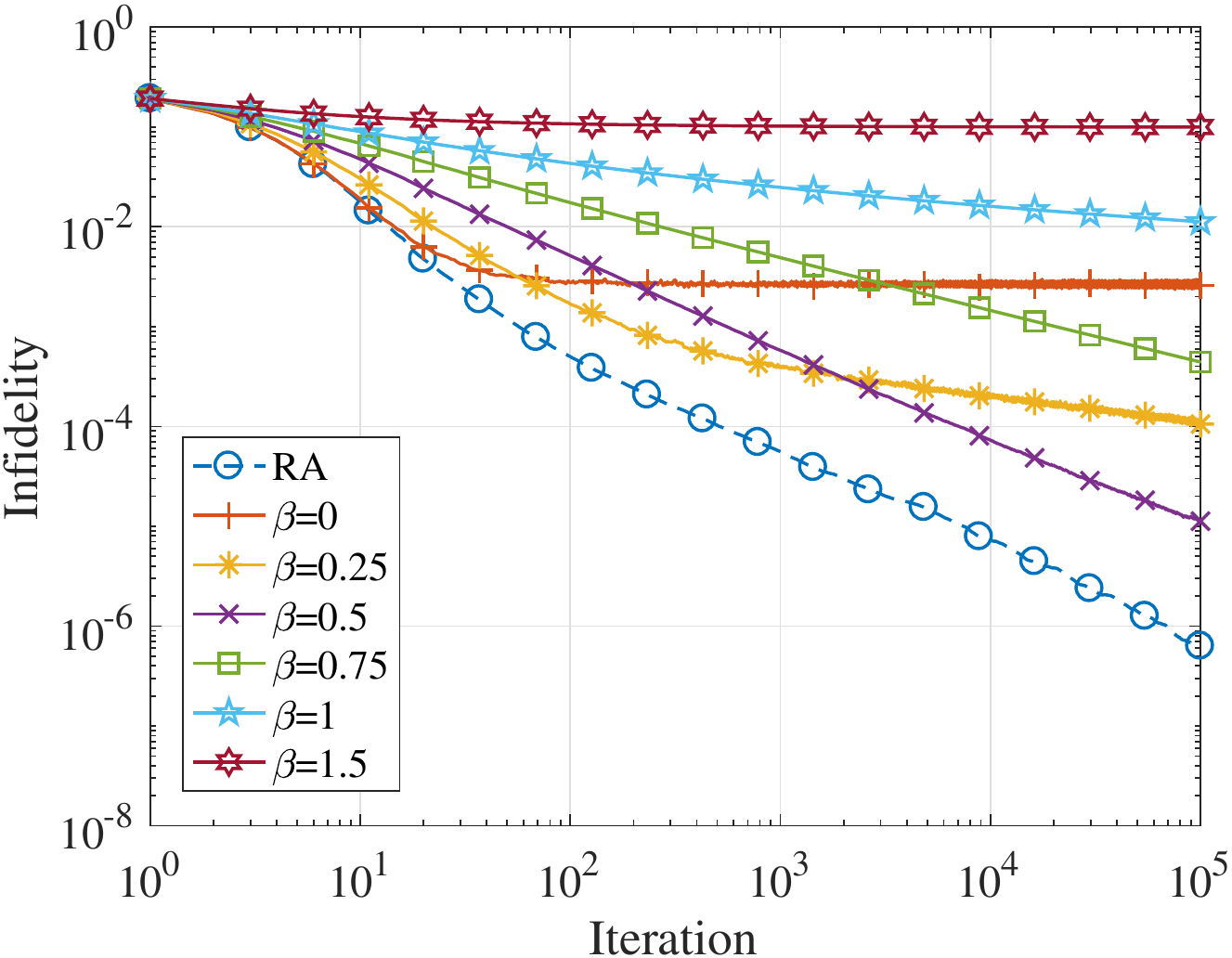}
\caption[]{Simulation results for MEG estimation for a single-qubit system and 100 shots per measurement. The infidelity of the proposed matrix exponential gradient (MEG) method is averaged over 1000 randomly generated quantum states and plotted versus the iteration number. The plot is for the running average case, as well as the variable learning rate $\eta_t = 0.5 t^{-\beta}$ for different values of $\beta$}
\label{fig:eta}
\end{figure}	

In addition to the matrix exponential gradient (MEG) estimator, the least squares (LS) method in \cite{qi_quantum_2013} and the diluted maximum likelihood (ML) in \cite{rehacek_diluted_2007} are also implemented and used for comparison in the setting of 1-, 2-, 3-, and 4-qubit systems. In these simulations, the learning rate of the MEG rule is 0.5. For the maximum likelihood method, the iteration step parameter $\epsilon$ (controlling the dilution) is taken to be 0.1. Since this value is much smaller than 1, it is guaranteed that after each internal iteration, the likelihood is increased as proved in \cite{rehacek_diluted_2007}. The number of internal ML iterations is chosen to be 10, which is a small number to reduce the total runtime of this method. In other words, for every new data point, we recalculate the ML estimate starting from the previous estimate using 10 iterations, and then evaluate the infidelity. An optimal setting would be a variable number of internal iterations that starts out large and decreases afterwards. However, it should be noted that in this work the objective is not optimizing the implementation of the ML, but to have the simplest implementation for comparison purpose. Additionally, we are interested more in the asymptotic behavior of the estimators. So, after a large number of data points, the estimate will be very near the true state. Consequently there will be no need to have a large number of ML internal iterations at that stage. The source code is publicly available \cite{sourcecode}. Figure \ref{fig:Results_multi} shows the infidelity versus the number of iterations for 1-, 2-, 3- and 4-qubits when the number of shots per measurement is taken to be 1000. Figure \ref{fig:Results_4} shows the performance for a 4-qubit system at different number of measurement shots. For 5-qubit systems, only the performance of MEG is assessed as shown in Figure \ref{fig:Results_5}.

\begin{figure*}[t]
\centering
\subfloat[10 shots per iteration]{\includegraphics[width=0.48\textwidth]{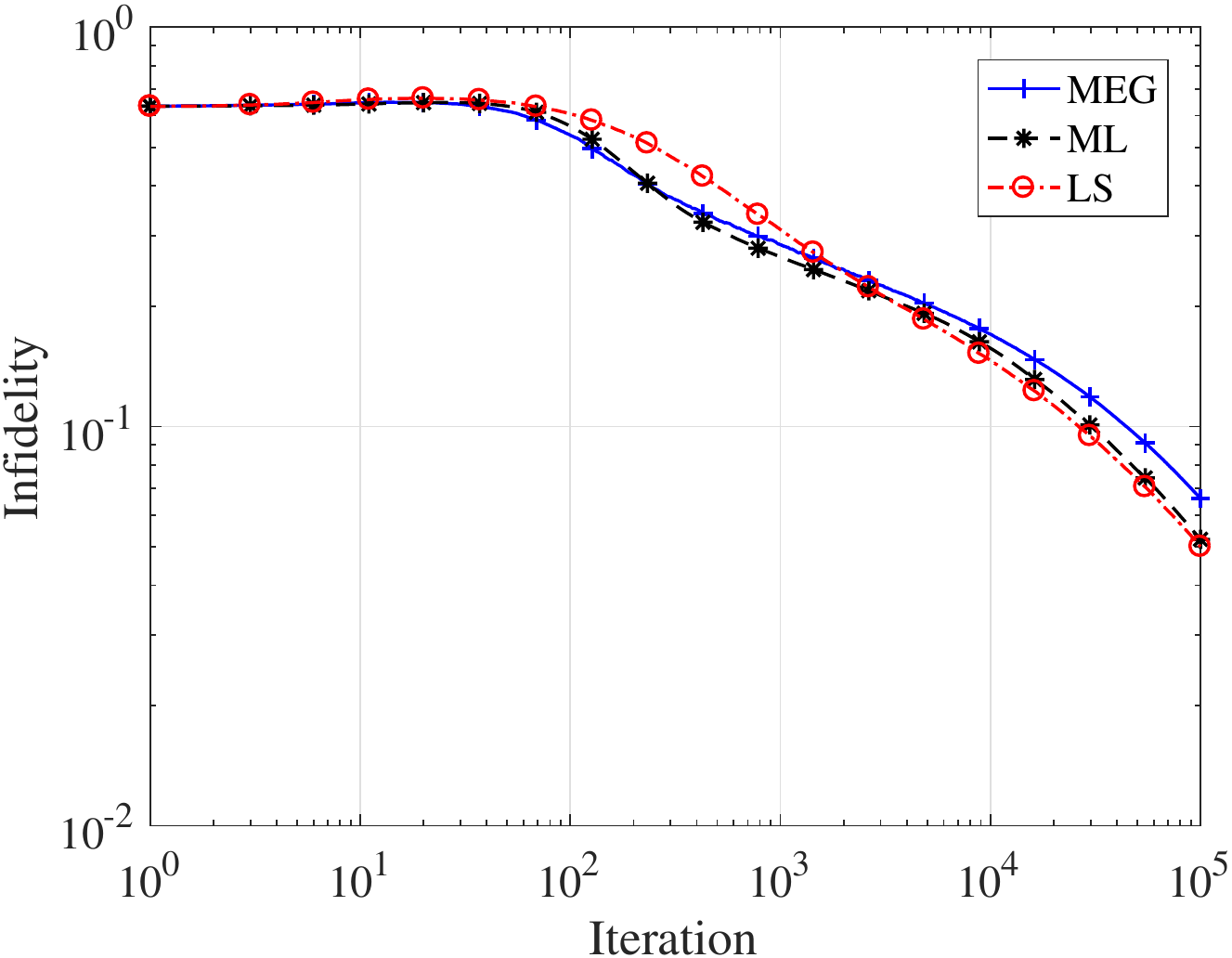}}
\subfloat[100 shots per iteration]{\includegraphics[width=0.48\textwidth]{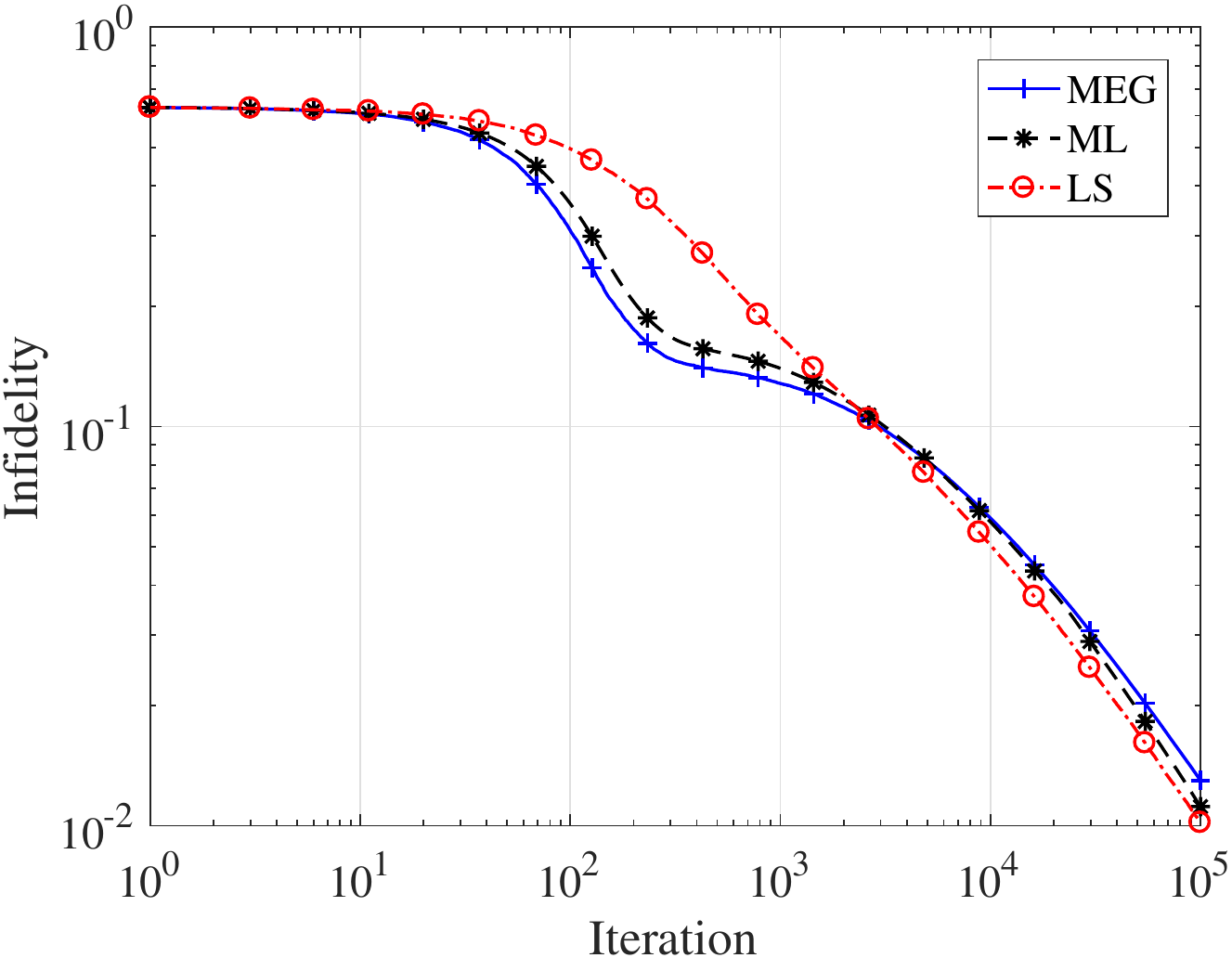}}\\
\subfloat[1000 shots per iteration]{\includegraphics[width=0.48\textwidth]{4_Q_3_nshots_loglog}}
\subfloat[10000 shots per iteration]{\includegraphics[width=0.48\textwidth]{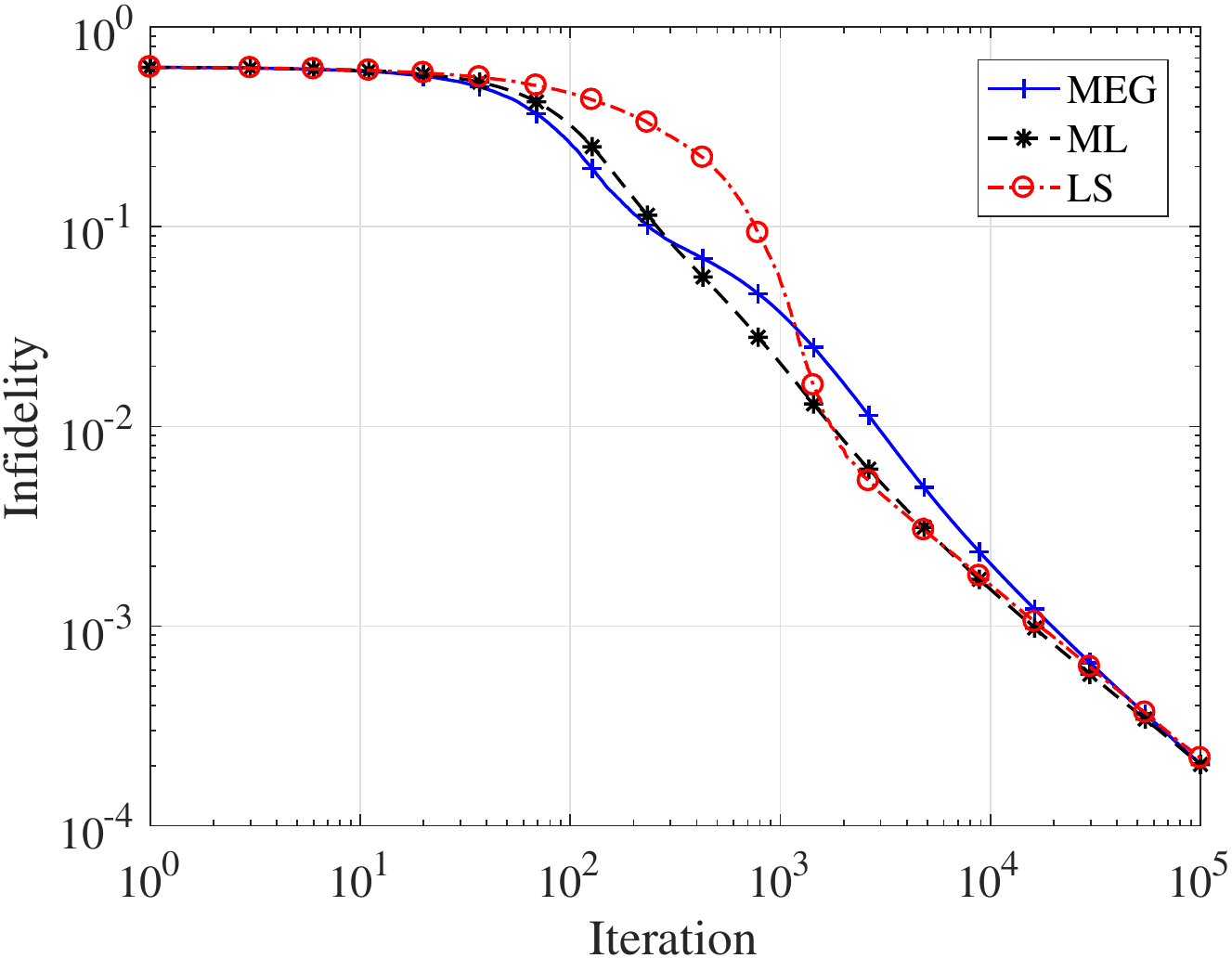}}

\caption[]{Simulation results for a four-qubit system. The infidelity is averaged over 1000 randomly generated quantum states and plotted versus the iteration number. The three lines correspond to the proposed matrix exponential gradient (MEG) method, maximum likelihood (ML) estimator and least-squares (LS) estimator. The number of shots per measurement is taken to be (a) 10 shots, (b) 100 shots, (c) 1000 shots, and (d) 10000 shots.}

\label{fig:Results_4}

\end{figure*}

\begin{figure}[t]
\centering
\includegraphics[width=0.55\textwidth]{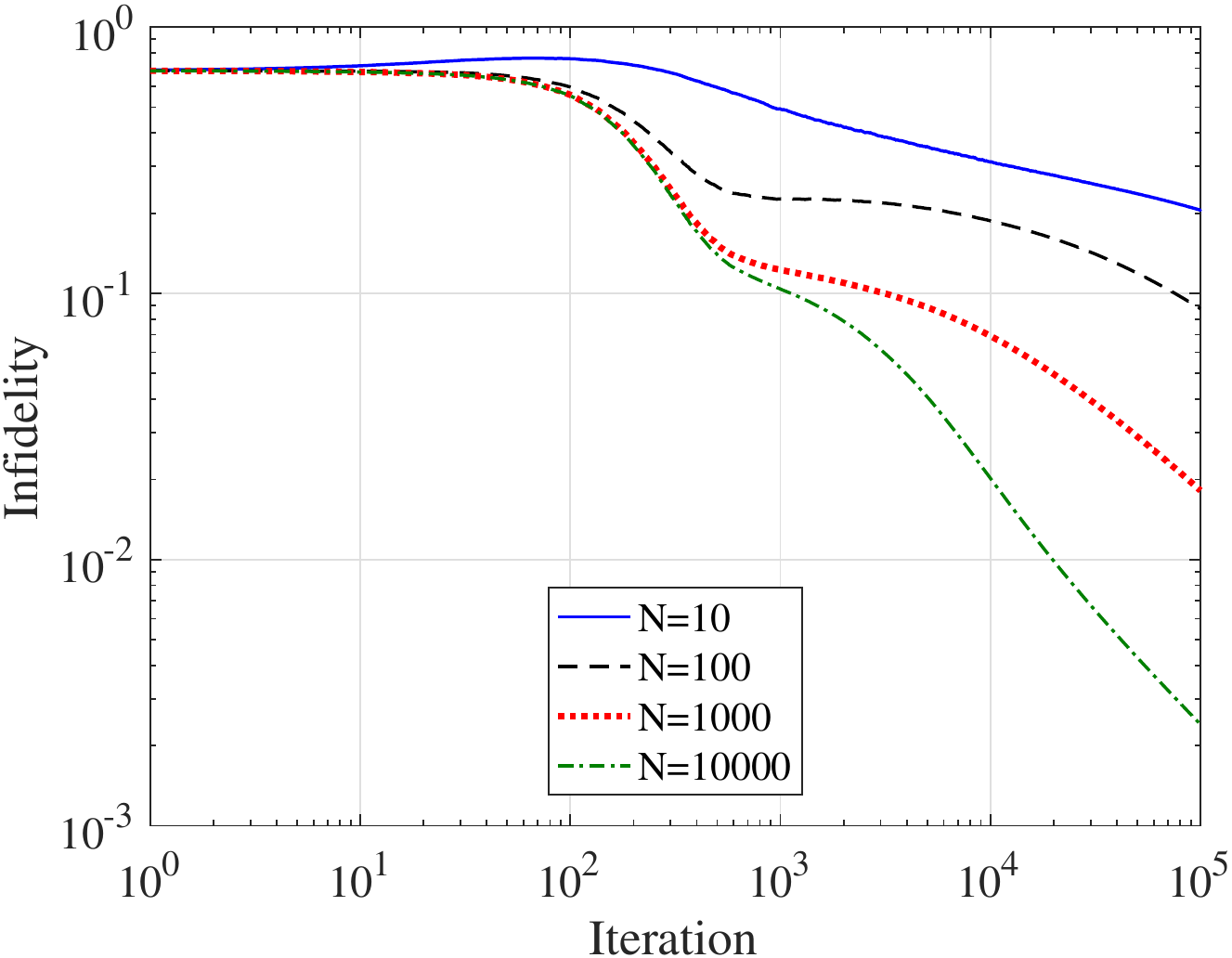}
\caption[]{Simulation results for a five-qubit system. The infidelity of the proposed matrix exponential gradient (MEG) method is averaged over 1000 randomly generated quantum states and plotted versus the iteration number. The number of shots per measurement is taken to be 10, 100, 1000, and 10000 shots.}
\label{fig:Results_5}
\end{figure}	
\subsection{Discussion \label{sec:discuss}}
The maximum likelihood method is a batch method that requires that the whole dataset is available for post-processing. So if a new measurement is done, the entire algorithm must be repeated again from the beginning. Additionally, the storage requirement of the data operators may be large, especially for multi-qubit systems. Our proposed method does not need to store all the data set, just the last averaged outcome for each measurement operator in the most sophisticated case. The same comparison applies to least-squares, which also acts on the whole batch of data and is not an online algorithm. 

An additional advantage of our algorithm is that it guarantees positivity of the estimated operator at all times. Least-squares and similar approaches are not guaranteed to produce a physical state unless a further step of projection back to the physical space is done. This forms an additional choice and overhead on the algorithm. Moreover, the use of running average allows using a constant learning rate. This solves the problem of having to evaluate the optimum learning rate at each time step.

Considering the accuracy of the estimate, the simulation results show that after a sufficient number of iterations, the MEG estimates converge to both the maximum-likelihood and least squares estimates which are considered the optimal estimators in batch processing systems. ML produces a point estimate for the model that maximizes the probability of the observed data, while LS minimizes the sum of squared errors due to observation noise. As the number of shots increase, the accuracy of all estimators gets better (i.e lower average infidelity for a given number of iterations) because the noise becomes less effective. On the other hand, as the number of qubits gets higher, more iterations are needed to achieve a low average infidelity. This is because at each iteration one basis is selected randomly for measurement. However, for high-dimensional systems there are many more bases that need to be covered to form a complete set ($d^2-1$ bases).  

As for complexity, maximum-likelihood scales as $O(d^4)$. This is because the bottleneck operation is calculating the gradient of the log-likelihood function $R=\sum_j{\frac{f_j}{N Pr_j}\Pi_j}$. For a complete set of measurement, at least $d^2-1$ measurement operators are needed, each of dimension $d\times d$. So this implies that calculating $R$ requires $O(d^4)$ complex multiplication operations. For the least-squares method, the complexity is $O(d^4)$ as discussed in \cite{qi_quantum_2013}. In this case the bottleneck operation is the matrix multiplication part $X^T Y$ of the estimation equation $\hat{\theta}=(X^{T}X)^{-1}X^T Y$. That is because again for a complete set of measurements we need at least $d^2-1$ operators, and thus $Y$ is of dimensions $(d^2-1)\times 1$, and $X$ is of dimensions $(d^2-1)\times (d^2-1)$. Finally, for the proposed method, the bottleneck is in calculating the matrix exponential. The complexity will depend on the particular way of implementation. The most common way is by performing eigendecomposition, followed by exponentiating the diagonal matrix of eigenvalues. In this case, the complexity is usually assumed to be $O(d^3)$ \cite{pan1999complexity, demmel2007fast}. It should be noted that the complexities discussed here are obtained per iteration, i.e for each update given a new data point. Table \ref{tab:complexity} summarizes these results. 

\begin{table}[!ht]
\centering
\caption{Summary of runtime complexities per iteration for the ML,LS, and MEG algorithms}
\label{tab:complexity}
\begin{tabular}{|c|c|}
\hline
\textbf{Algorithm} & \textbf{Runtime}\\ \hline
ML                 & $O(d^4)$    \\ \hline
LS                  & $O(d^4)$    \\ \hline
MEG               & $O(d^3)$    \\ \hline
\end{tabular}
\end{table}

In order to verify the claim that MEG should have the fastest performance, the execution times per 1 iteration were recorded in the simulation for the three methods. Figure \ref{fig:complexity} shows the average of these execution times. It is clear that as the number of qubits increases, MEG has the least runtime compared to the other two methods. 
\begin{figure}[t]
\includegraphics[width=0.55\textwidth]{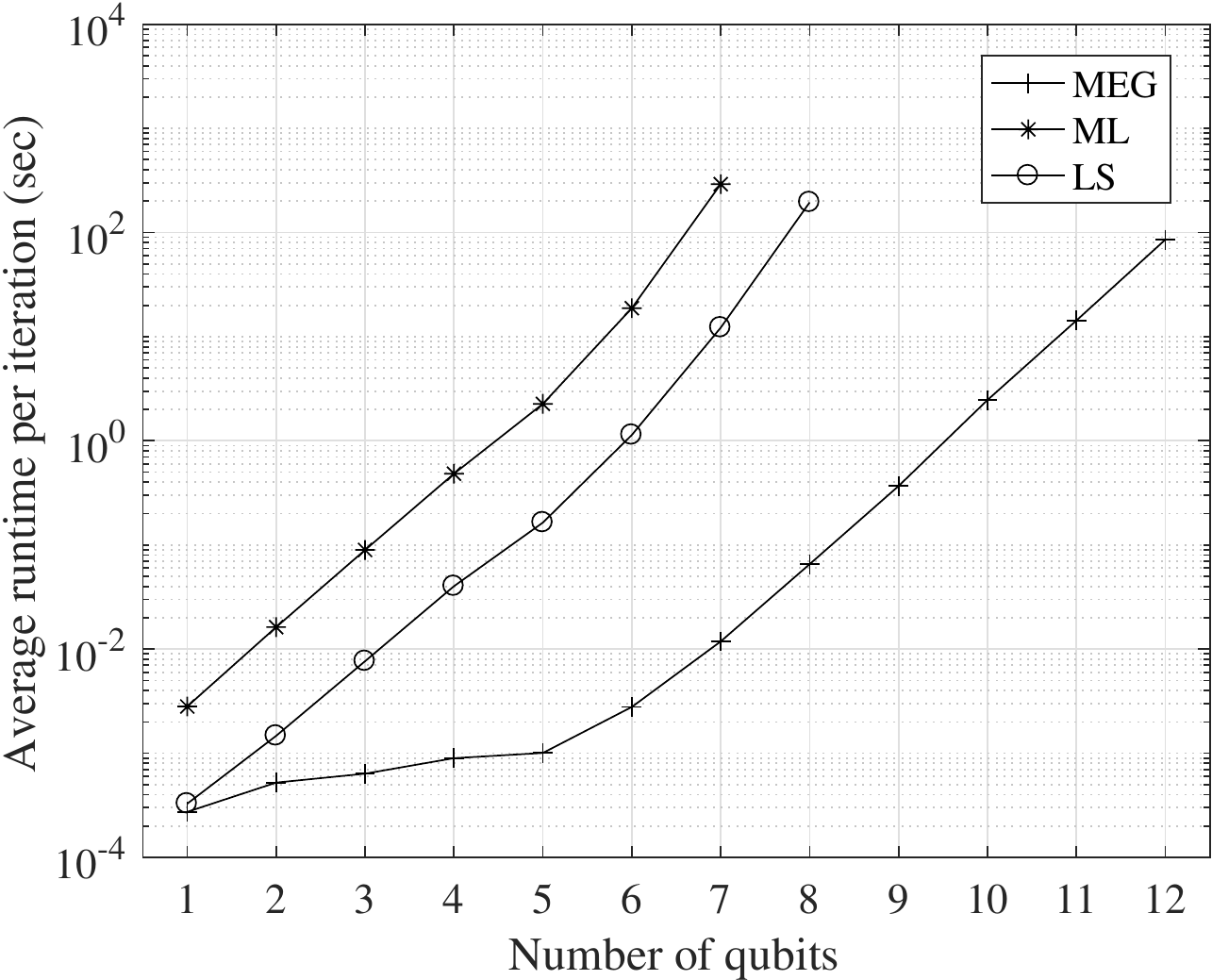}
\caption{The average runtime of the update step for the maximum-likelihood (ML), least-squares (LS), and matrix-exponentiated gradient (MEG) methods, measured for increasing number of qubits.}
\label{fig:complexity}
\end{figure} 
\section{Conclusion}
In this paper, we introduced the idea of using the running average on the noisy measurements together with the MEG update rule to construct a fast and simple online quantum state estimator. However, there are still some points to consider in the future. First, we considered only fixed measurements, but having adaptive measurements could further improve the performance. Also, it would be interesting to test these ideas while embedded in a real experiment. Finally, it would be interesting to explore other possible machine learning techniques in the classical literature, and investigate their applicability in the quantum setting. In particular, proving convergence for the projected-gradient method as another online estimation algorithm would be worth considering. Appendix \ref{sec:pgd} gives more details on this point.

\paragraph*{Acknowledgments:}
AY is supported by an Australian Government Research Training Program Scholarship. 
MT and CF acknowledge Australian Research Council Discovery Early Career Researcher Awards, projects No.\ DE160100821 and DE170100421, respectively. 
This research is also supported in part by the ARCLab facility at UTS.

\appendix

\section{Auxiliary lemmas}
\label{apndx:lem}
In this Appendix, we present some auxiliary lemmas needed for some proofs. We will start by stating the following lemma \cite{tsuda_matrix_2005}, which is proved as Lemma 1 in \cite{helmbold1997comparison}. 
\begin{lemma}
\label{lem:log}
Let $0\le q \le 1$, then for any $p$,
\begin{align}
	\log\left(1-q\left(1-\exp\left(p\right)\right)\right)\le p q+\frac{p^2}{8}.
\end{align}

\end{lemma}
Next, we state the Golden-Thompson inequality \cite{golden1965lower, thompson1965inequality}.
\begin{lemma}[Golden-Thompson Inequality]
\label{lem:GT}
Let $A$ and $B$ be two Hermitian matrices, then
\begin{align}
	\tr(\exp(A+B))\leq \tr(\exp(A)\exp(B)).
\end{align}

\end{lemma}
The following result is presented as Lemma 2.1 in \cite{tsuda_matrix_2005}.
\begin{lemma}[Jensen's Inequality]
\label{lem:Jn}
Let  $0 \le A \le I$, and $x,y \ge 0$, then 
\begin{align}
	\exp(x A+ y (I-A)) &\leq \exp(x)A+\exp(y)(I-A)\\
	&= I \exp(y)+(\exp(x)-\exp(y))A.
\end{align}

\end{lemma}
Finally, we state the following lemma relating convergence in mean to convergence in probability. 

\begin{lemma}
\label{lem:prob}
Given a sequence of positive random variables $Z_t$,
\begin{align}
	\lim_{t\to\infty}{\E\{Z_t\}}=0 \implies \forall \delta>0, \lim_{t\to\infty}{\Pr\left\{Z_t \le\delta\right\}}=1.
\end{align} 
\end{lemma}

\begin{proof}
The statement 
\begin{align}
	\lim_{t\to\infty}{\E\{Z_t\}}=0.
\end{align}
is equivalent to the statement
\begin{align}
	\forall \epsilon>0,\delta >0, \exists T_\delta: \forall t>T_\delta, \E\{Z_t\}\le\delta\epsilon.
\end{align}
Now, Markov inequality states that for a non-negative random variable $X$, 
\begin{align}
	\pr\{X\ge a\}\le\frac{\E\{X\}}{a}.
\end{align}
So, the previous definition of the limit becomes
\begin{align}
	\forall \epsilon>0,\delta >0, \exists T_\delta: \forall t>T_\delta, \Pr\left\{Z_t\ge\delta\right\} \le \epsilon,
\end{align}
or,
\begin{align}
	\forall \epsilon>0,\delta >0, \exists T_\delta: \forall t>T_\delta, \Pr\left\{Z_t < \delta\right\} \ge 1-\epsilon.
\end{align}
Writing back as a limit, the expression becomes
\begin{align}
	\forall \delta>0, \lim_{t\to\infty}{\Pr\left\{Z_t < \delta\right\}}=1,
\end{align}
which is the definition of convergence in probability. 
\end{proof}

\section{Additional proofs}\label{appendix:proofs}
This appendix lists proofs of some lemmas that were not given in the main text. The proofs follow the same methods in \cite{tsuda_matrix_2005}, generalized to work with the quantum case. 
\tocless\subsection{Proof of Lemma \ref{lem:log_Zt}}
\label{apndx:lem:log_Zt}
\logzt*
\begin{proof}
Recall that
\begin{align}
	\delta_t=-2\eta(\tr(\hat{\rho}_t X_t)-\hat{y}_t).
\end{align}
Applying Golden-Thompson inequality in Lemma \ref{lem:GT}, we get
\begin{align}
	\log(\tr(\exp(\log(\hat{\rho}_t)+\delta_t X_t))) &\leq \log(\tr(\hat{\rho}_t \exp(\delta_t X_t)))\\
	&=\log\left(\tr\left(\hat{\rho}_t\exp(-\delta_t)\exp\left(2 \delta_t \frac{X_t + I}{2}\right)\right)\right)\\
	&=- \delta_t + \log\left(\tr\left(\hat{\rho}_t\exp\left(2 \delta_t \frac{X_t + I}{2}\right)\right)\right).
\end{align}
Applying Jensen's inequality in Lemma \ref{lem:Jn} by choosing $A=\frac{X_t+ I}{2}$, $x=2\delta_t$, and $y=0$; we get that
\begin{align}
	\log(\tr(\exp(\log(\hat{\rho}_t)+\delta_t X_t))) &\le - \delta_t + \log\left(\tr\left(\hat{\rho}_t\left(I-\left(1-\exp(2\delta_t)\right)\frac{X_t + I}{2}\right)\right)\right) \\
	&= - \delta_t + \log\left(1-\left(1-\exp(2\delta_t)\right)\frac{\tr\left(\hat{\rho}_t(X_t + I)\right)}{2}\right)\\
	&=- \delta_t + \log\left(1-\left(1-\exp(2\delta_t)\right)\frac{\tr\left(\hat{\rho}_t X_t\right) + 1}{2}\right). 
\end{align}
Applying now the log identity in Lemma \ref{lem:log}, with $p=2\delta_t$, and $q=\frac{\tr\left(\hat{\rho}_t X_t\right) +1}{2}$, we obtain 
\begin{align}
	\log(\tr(\exp(\log(\hat{\rho}_t)+\delta_t X_t))) &\le - \delta_t +\frac{(2\delta_t)^2 }{8}+2\delta_t\frac{\tr\left(\hat{\rho}_t X_t\right) +1 }{2}\\
	&=\frac{\delta_t^2 }{2} +\delta_t\tr\left(\hat{\rho}_t X_t\right),
\end{align}
which completes the proof.
\end{proof}
\tocless\subsection{Proof of Lemma \ref{lem:Lt_inequ}}
\label{apndx:lem:Lt_inequ}
\ltinequ*
\begin{proof}
We start with calculating the right hand side, 
\begin{align}
	D(\sigma||\hat{\rho}_t)-D(\sigma||\hat{\rho}_{t+1})&=\tr(\sigma\log(\sigma)-\sigma\log(\hat{\rho}_t))-\tr(\sigma\log(\sigma)-\sigma\log(\hat{\rho}_{t+1}))\\
	&=-\tr(\sigma\log(\hat{\rho}_t))+\tr(\sigma\log(\hat{\rho}_{t+1}))\\
	&=-\tr(\sigma\log(\hat{\rho}_t))+\tr(\sigma\log(\hat{\rho}_{t}))+\tr(\sigma\delta_t X_t)\nonumber \\
	&-\tr(\sigma \log(\tr(\exp(\log(\hat{\rho}_t)+\delta_t X_t))))\\
	&=\delta_t \tr(\sigma X_t) - \log(\tr(\exp(\log(\hat{\rho}_t)+\delta_t X_t)))
\end{align}
Applying Lemma \ref{lem:log_Zt} we get
\begin{align}
	D(\sigma||\hat{\rho}_t)-D(\sigma||\hat{\rho}_{t+1})&\ge \delta_t \tr(\sigma X_t) - \frac{\delta_t^2 }{2} -\delta_t\tr\left(\hat{\rho}_t X_t\right)\\
	&=-2\eta(\tr(\hat{\rho}_tX_t)-\hat{y}_t)(\tr(\sigma X_t)-\tr(\hat{\rho}_t X_t))-2\eta^2(\tr(\hat{\rho}_tX_t)-\hat{y}_t)^2\\
	&=-2\eta(\tr(\hat{\rho}_tX_t)-\hat{y}_t)(\tr(\sigma X_t)-\hat{y}_t+\hat{y}_t-\tr(\hat{\rho}_t X_t)) \nonumber \\
	&-2\eta^2(\tr(\hat{\rho}_tX_t)-\hat{y}_t)^2\\
	&=\left(2\eta-2\eta^2\right)(\tr(\hat{\rho}_tX_t)-\hat{y}_t)^2-2\eta(\tr(\sigma X_t)-\hat{y}_t)(\tr(\hat{\rho}_t X_t)-\hat{y}_t)\\
	&\ge\left(2\eta-2\eta^2\right)L_t(\hat{\rho}_t)-2\eta\sqrt{L_t(\hat{\rho}_t)L_t(\sigma)}\\
	&=\left(\sqrt{\eta-2\eta^2}\sqrt{L_t(\hat{\rho}_t)}-\sqrt{\frac{\eta^2}{\eta-2\eta^2}}\sqrt{L_t(\sigma)}\right)^2 \nonumber \\
	&+\eta L_t(\hat{\rho}_t)-\frac{\eta^2}{\eta-2\eta^2}L_t(\sigma).\end{align}
If we now choose $\eta-2\eta^2>0$, then the square roots in the last expression are real valued. As a result,
\begin{align}
	D(\sigma||\hat{\rho}_t)-D(\sigma||\hat{\rho}_{t+1})\ge \eta L_t(\hat{\rho}_t)-\frac{\eta}{1-2\eta}L_t(\sigma),
\end{align} 
and the learning factor $\eta$ must satisfy
\begin{align}
	0<\eta<\frac{1}{2},
\end{align}
which completes the proof of the lemma. To account for noiseless measurements, $\hat{y}_t$ is just replaced by $y_t$. 
\end{proof}


\tocless\subsection{Proof of Theorem \ref{lem:weak_conv_noisy}}
\label{appndx:lem:weak_conv_noisy}
\weakconv*
\begin{proof}
We know from Lemma \ref{lem:Lt_inequ} that,
\begin{align}
	\eta_t L_t(\hat{\rho}_t) - \frac{\eta_t}{1-2\eta_t} L_t(\rho) \le D(\rho||\hat{\rho}_t)-D(\rho||\hat{\rho}_{t+1}).
\end{align}
Taking the expectation with respect to $y_t$ followed by the the expectation with respect to $X_t$ we get,
\begin{align}
	\eta_t\E_t\{L_t(\hat{\rho}_t)\} -\frac{\eta_t}{1-2\eta_t}\E_t\{L_t(\rho)\} \le D(\rho||\hat{\rho}_t) - \E_t\{D(\rho||\hat{\rho}_{t+1})\}.
\end{align}
Applying Lemma \ref{lem:E_Lt_noisy}, we get
\begin{align}
	\eta_t\frac{d}{d^2-1}\left(\norm{\hat{\rho}_t-\rho}_F^2 + \frac{d-\norm{\rho}_F^2}{N}\right) - \frac{\eta_t}{1-2\eta_t} \frac{d}{d^2-1}\left(\frac{d-\norm{\rho}_F^2}{N}\right) \le D(\rho||\hat{\rho}_t) - \E_t\{D(\rho||\hat{\rho}_{t+1})\}.
\end{align}
Simplifying this expression and taking the expectation with respect to all previous time instants we get
\begin{align}
	\eta_t\E\{\norm{\hat{\rho}_t-\rho}_F^2\}  - \frac{2\eta_t^2}{1-2\eta_t}\left(\frac{d-\norm{\rho}_F^2}{N}\right) \le \frac{d^2-1}{d}\E\{D(\rho||\hat{\rho}_t) - D(\rho||\hat{\rho}_{t+1})\}.
\end{align}
The second term on the left hand side depends on the purity of the true state, and it can be bounded to become
\begin{align}
	\eta_t\E\{\norm{\hat{\rho}_t-\rho}_F^2\}  - \frac{2\eta_t^2}{1-2\eta_t}\left(\frac{d^2-1}{Nd}\right) \le \frac{d^2-1}{d}\E\{D(\rho||\hat{\rho}_t) - D(\rho||\hat{\rho}_{t+1})\}.
\end{align}
Selecting the learning rate to be 
\begin{align}
	\eta_t = \frac{1}{2}\frac{\E\{\norm{\hat{\rho}_t-\rho}_F^2\}}{\E\{\norm{\hat{\rho}_t-\rho}_F^2\}+2\left(\frac{d^2-1}{Nd}\right)},
\end{align}
then summing up the inequality over different time steps yields
\begin{align}
	\sum_{t=1}^{T}{\frac{1}{4}\frac{\E\left\{\norm{\hat{\rho}_t-\rho}_F^2\right\}^2}{\E\left\{\norm{\hat{\rho}_t-\rho}_F^2\right\} +2\left(\frac{d^2-1}{Nd}\right)}} &\le \frac{d^2-1}{d}\left(\E\{D(\rho||\hat{\rho}_1) - D(\rho||\hat{\rho}_{T+1})\}\right)\\
	&\le \frac{d^2-1}{d}D(\rho||\hat{\rho}_1),
\end{align}
where $\E\{D(\rho||\hat{\rho}_1)\}=D(\rho||\hat{\rho}_1)$ because $\hat{\rho}_1$ and $\rho$ are independent of $X_t$ and $\hat{y}_t$ for any $t$. Now taking the limit as $T\to\infty$ we get
\begin{align}
	\sum_{t=1}^{\infty}{\frac{\E\left\{\norm{\hat{\rho}_t-\rho}_F^2\right\}^2}{\E\left\{\norm{\hat{\rho}_t-\rho}_F^2\right\} +2\left(\frac{d^2-1}{Nd}\right)}} \le 4\left(\frac{d^2-1}{d}\right)D(\rho||\hat{\rho}_1).
\end{align}
Since, the left-hand side of the inequality is constant, then the series on the right hand side must converge. Consequently using the divergence test, 
\begin{align}
	\lim_{t\to\infty}{\frac{\E\left\{\norm{\hat{\rho}_t-\rho}_F^2\right\}^2}{\E\left\{\norm{\hat{\rho}_t-\rho}_F^2\right\} +2\left(\frac{d^2-1}{Nd}\right)}}=0.\label{equ:lim}
\end{align}
Assume that
\begin{align}
	\lim_{t\to\infty}{\E\left\{\norm{\hat{\rho}_t-\rho}_F^2\right\}=K>0},
\end{align}
then 
\begin{align}
	\lim_{t\to\infty}{\frac{\E\left\{\norm{\hat{\rho}_t-\rho}_F^2\right\}^2}{\E\left\{\norm{\hat{\rho}_t-\rho}_F^2\right\} +2\left(\frac{d^2-1}{Nd}\right)}}=\frac{K^2}{K+2\left(\frac{d-\norm{\rho}_F^2}{N}\right)}\ne 0,
\end{align}
which contradicts the condition in (\ref{equ:lim}). This means that it must be the case that 
\begin{align}
	\lim_{t\to\infty}{\E\left\{\norm{\hat{\rho}_t-\rho}_F^2\right\}}=0.
\end{align}
Now we can apply Lemma \ref{lem:prob} on the random variable $Z_t = \norm{\hat{\rho}_t-\rho}_F^2$  to conclude that 
\begin{align}
	\forall \delta>0, \lim_{t\to\infty}{\Pr\left\{\norm{\hat{\rho}_t-\rho}_F^2\le\delta\right\}}=1.
\end{align}
Therefore, the estimate $\hat{\rho}_t$ converges in probability to the true state $\rho$.
\end{proof}

\section{Overview on the diluted maximum likelihood method}

In this appendix an overview on the diluted maximum likelihood \cite{rehacek_diluted_2007} is given. The maximum likelihood method of quantum estimation is based on trying to find the state $\hat{\rho}$ that maximizes the log-likelihood function
\begin{align}
	\log(\mathcal{L})=\sum_j{f_j\log(\tr(\hat{\rho} \Pi_j))},
\end{align}
where $f_j$ is the relative frequency of outcome $j$ described by the POVM set $\{\Pi_j\}$. This can be achieved by doing iterations in the form
\begin{align}
	\hat{\rho}_{k+1} = R \hat{\rho}_k R,
\end{align}
where 
\begin{align}
	R=\sum_j{\frac{f_j}{\tr(\hat{\rho} \Pi_j)}\Pi_j}.
\end{align}
This is the $R\rho R$ algorithm. However, there is no guarantee that this form of update equation will generally converge. So, a modification on the form of the update equation is done to become, 
\begin{align}
	\hat{\rho}_{t+1} = \frac{(I+\epsilon R)\hat{\rho}_t(I+\epsilon R)}{\tr\left((I+\epsilon R)\hat{\rho}_t(I+\epsilon R)\right)}.
\end{align}
This is called the diluted maximum likelihood because it ``dilutes'' $R$ by mixing it with the identity operator $I$. The step parameter $\epsilon$ can be chosen arbitrarily and can be constant or adaptive. When $\epsilon\to\infty$, the iterations reverts back to the $R\rho R$ form. Choosing a value of $\epsilon\ll 1$ ensures that after each iteration the likelihood function is non-decreasing. On the other hand, a higher value of $\epsilon$ is better in terms of speed of convergence. A fewer number of iterations is required to achieve the same accuracy compared to a low value of $\epsilon$. 

An important thing to notice is that Pauli operators $\{X^{(i)}\}_{i=1}^{d^2-1}$ do not form a set of POVM. But fortunately it is easy to construct a set of POVM out of the Pauli operators, by taking all the ``up/down'' projectors normalized. So, the POVM set becomes $\{\frac{1}{d^2-1}\Pi_{\uparrow}^{(i)},\frac{1}{d^2-1}\Pi_{\downarrow}^{(i)}\}_{i=1}^{d^2-1}$.

Finally, maximum likelihood is a batch algorithm. So, the iterations are repeatedly run on a set of data. To modify this algorithm to become online, the iterations must be performed on the dataset after each new data point obtained. 

\section{Overview on the least-squares method}
\label{sec:ls}
This appendix gives a short brief on the least-squares method for quantum estimation. More details are given in \cite{qi_quantum_2013}. The basic idea is to construct a parametric model of the state in the form
\begin{align}
	\rho = \frac{I}{d} + \sum_{i=1}^{d^2-1}{\theta_i U_i},
\end{align} 
where the set $\{U_i\}_{i=1}^{d^2-1}$ are some Hermitian basis, and $\theta_i = \tr(\rho U_i)$ are the parameters. The problem of quantum tomography then becomes trying to estimate this parameter vector given the measurement dataset. So given a set of measurement operators represented using their ``up/down'' projectors in the form $\{\Pi_j\}_{j=1}^{2d^2-2}=\{\Pi_{\uparrow}^{(i)},\Pi_{\downarrow}^{(i)}\}_{i=1}^{d^2-1}$, the associated probabilities are
\begin{align}
	p_j = \tr(\Pi_j \rho) = \frac{1}{d} + \sum_{i=1}^{d^2-1}{\theta_i\tr(\Pi_j U_i)}.
\end{align}
These probabilities can be obtained experimentally but with some errors due to performing finite number of shots. The noisy data is denoted by $\hat{p}_j$. By defining the matrix expansion of the projectors in terms of the chosen basis $X_{i,j} = \tr(\Pi_i U_j)$, the dependent vector with components $Y_j=\hat{p}_j-\frac{1}{d}$, the model can be rewritten as 
\begin{align}
	Y=X\theta + e,
\end{align}
where $e$ is the error vector which converges to a normal distributed random variable at the limit of very large number of measurements. The optimal parameter is defined to minimize the sum of squared errors as
\begin{align}
	\hat{\theta}_{LS} = \text{argmin}_{\theta} (Y-X\theta)^T(Y-X\theta),
\end{align}
and the solution of this optimization problem is
\begin{align}
	\hat{\theta}_{LS} = (X^T X)^{-1}X^T Y.
\end{align} 
After the estimation of the unknown parameter, the quantum state is reconstructed as 
\begin{align}
	\hat{\rho} = \frac{I}{d} + \sum_{i=1}^{d^2-1}{\hat{\theta}_i U_i}.
\end{align}
The reconstructed state might be generally unphysical due to non-positivity of the estimate. So, in this case the state must be projected back to the physical space. One way to do is to redistribute the negative eigenvalues over all other eigenvalues, until there are no more negative eigenvalues. It can be shown \cite{smolin_efficient_2012} that this is an optimal projection method, in the sense that the projected state is nearest to the unphysical state in terms of the Frobenuis norm. This method is batch, but can be adapted to become online by doing the whole procedure of estimation and projection after each new data point obtained.

\section{Comparison with online projected-gradient descent}
\label{sec:pgd}

\begin{figure*}[t]

\centering
\subfloat[PGD with learning rate of 0.5]{\includegraphics[width=0.48\textwidth]{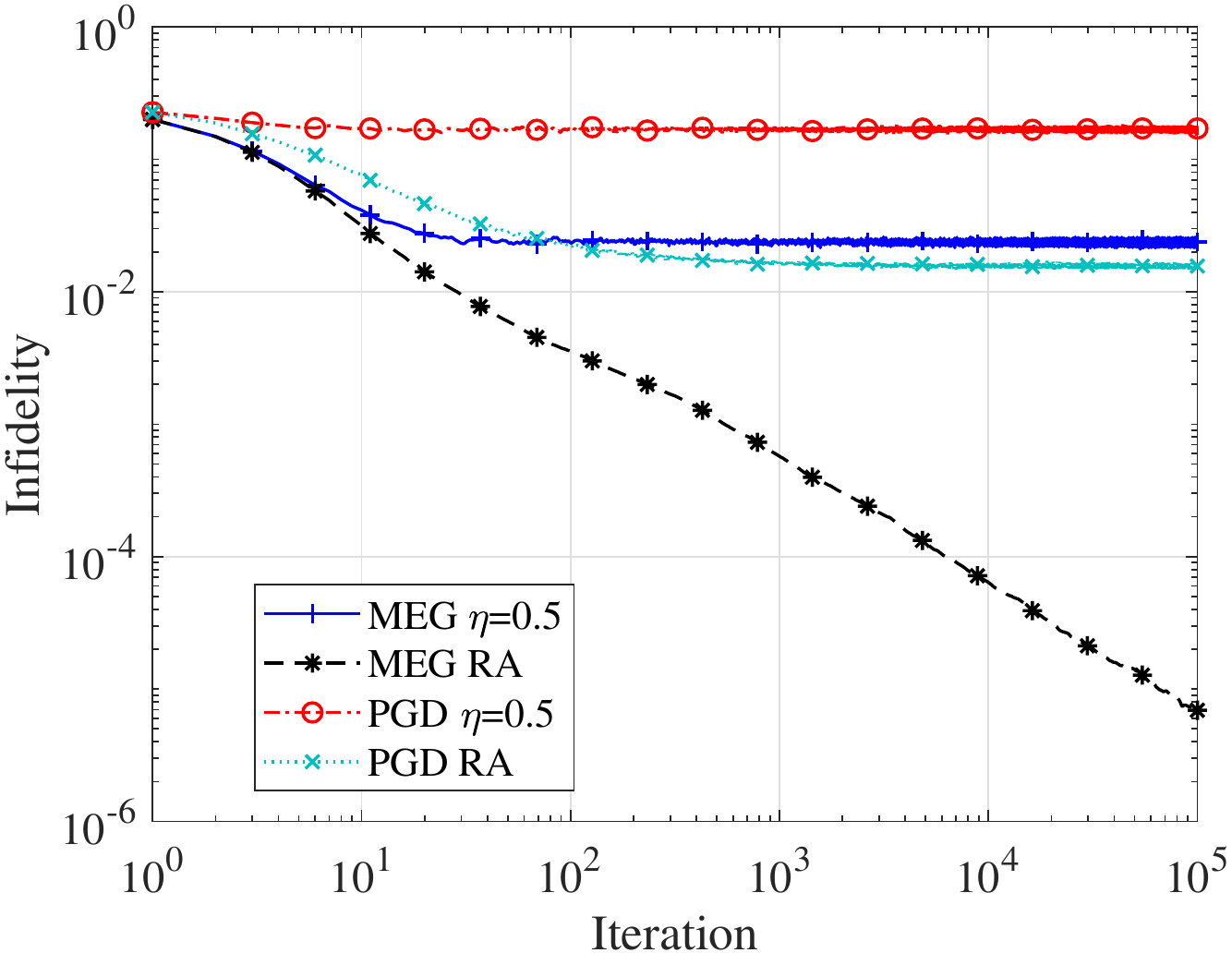} \label{fig:Results_PGD1} } 
\subfloat[PGD with learning rate of 0.05]{\includegraphics[width=0.48\textwidth]{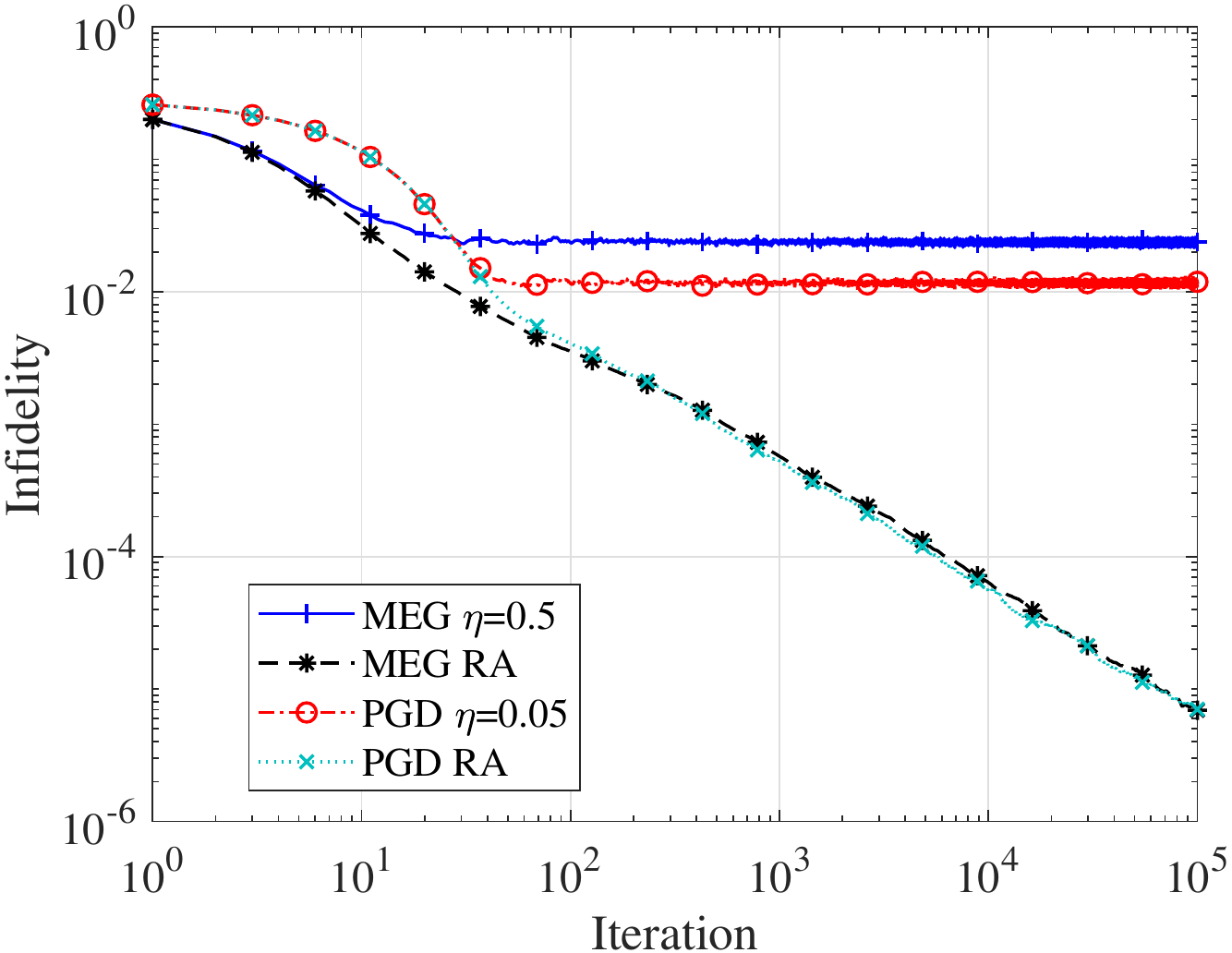} \label{fig:Results_PGD2}} 

\caption[]{Simulation results for comparing MEG with constant learning rate of 0.5 and MEG with running averages to the projected-gradient descent method (PGD) with constant learning rate as well as running averages for learning rate of (a) 0.5 and (b) 0.05. The infidelity is averaged over 1000 randomly generated quantum states and plotted versus the iteration number. The number of shots per measurement is taken to be 10 shots. }

\label{fig:Results_PGD}

\end{figure*}
The projected-gradient method (PGD) was proposed in \cite{bolduc_projected_2017} in the batch setting. We implemented this method in the online setting, and investigated its performance numerically in this case. The estimate at time iteration $t+1$ is given by
\begin{align}
	\hat{\rho}_{t+1} = \mathcal{P}\left(\hat{\rho}_t - \eta\nabla L_t(\hat{\rho}_t)\right),
\end{align}
where $\eta$ is the learning rate, $L_t$ is the loss function as defined previously, and $\mathcal{P}$ denotes projection into the physical space as discussed in Appendix \ref{sec:ls}. It turns out the performance is highly dependent on choosing the learning rate, and generally seems very similar to MEG when the learning rate is chosen low. In Figure \ref{fig:Results_PGD1}, we compared MEG with constant learning rate, MEG with running average (RA), PGD with constant learning rate and PGD with running average. The same value of $\eta=0.5$ was used in the four of them. This plot shows that the MEG method has better convergence. On the other hand, by changing the step size of the PGD methods to 0.05, we see that both methods seem to have similar performance for the running average case after significant number of iterations as shown in Figure \ref{fig:Results_PGD2}. This makes it very difficult to give a fair comparison with MEG. Note that even for MEG we were not particularly concerned with finding the optimal learning rate\,---\,we simply tested a few learning rates that are compatible with the limitations given by the convergence proof. However, faster (but not provable) convergence might be possible, as it often is, if we go outside that range. For PGD we simply do not know what the restrictions on the learning rate are so that the algorithm still provably converges with similar parameters as MEG. Thus, it will be interesting as a future work to look into the convergence of the PGD method. 

\bibliographystyle{alpha}
\bibliography{library,library_mt} 

\newcommand{\etalchar}[1]{$^{#1}$}
\begin{thebibliography}{GKCC07}

\bibitem[Aar07]{aaronson2007learnability}
Scott Aaronson.
\newblock The learnability of quantum states.
\newblock In {\em Proceedings of the Royal Society of London A: Mathematical,
  Physical and Engineering Sciences}, volume 463, pages 3089--3114. The Royal
  Society, 2007.

\bibitem[Aar18]{aaronson2018shadow}
Scott Aaronson.
\newblock Shadow tomography of quantum states.
\newblock In {\em Proceedings of the 50th Annual ACM SIGACT Symposium on Theory
  of Computing}, pages 325--338. ACM, 2018.

\bibitem[ACHN18]{aaronson_online_2018}
Scott Aaronson, Xinyi Chen, Elad Hazan, and Ashwin Nayak.
\newblock Online {Learning} of {Quantum} {States}.
\newblock {\em arXiv:1802.09025 [quant-ph]}, February 2018.
\newblock arXiv: 1802.09025.

\bibitem[ADA13]{adam2013applications}
Martin ADAM.
\newblock {\em Applications of Unitary k-designs in Quantum Information
  Processing}.
\newblock PhD thesis, Masarykova univerzita, Fakulta informatiky, 2013.

\bibitem[BKGL17]{bolduc_projected_2017}
Eliot Bolduc, George~C. Knee, Erik~M. Gauger, and Jonathan Leach.
\newblock Projected gradient descent algorithms for quantum state tomography.
\newblock {\em npj Quantum Information}, 3(1):44, October 2017.

\bibitem[DDH07]{demmel2007fast}
James Demmel, Ioana Dumitriu, and Olga Holtz.
\newblock Fast linear algebra is stable.
\newblock {\em Numerische Mathematik}, 108(1):59--91, 2007.

\bibitem[FGLE12]{flammia_quantum_2012}
Steven~T. Flammia, David Gross, Yi-Kai Liu, and Jens Eisert.
\newblock Quantum tomography via compressed sensing: error bounds, sample
  complexity and efficient estimators.
\newblock {\em New Journal of Physics}, 14(9):095022, 2012.

\bibitem[GCC16]{granade_practical_2016}
Christopher Granade, Joshua Combes, and D.~G. Cory.
\newblock Practical {Bayesian} tomography.
\newblock {\em New Journal of Physics}, 18(3):033024, 2016.

\bibitem[GKCC07]{globerson_exponentiated_2007}
Amir Globerson, Terry~Y. Koo, Xavier Carreras, and Michael Collins.
\newblock Exponentiated gradient algorithms for log-linear structured
  prediction.
\newblock In {\em In {Proc}. {ICML}}, pages 305--312, 2007.

\bibitem[GLF{\etalchar{+}}10]{gross_quantum_2010}
David Gross, Yi-Kai Liu, Steven~T. Flammia, Stephen Becker, and Jens Eisert.
\newblock Quantum {State} {Tomography} via {Compressed} {Sensing}.
\newblock {\em Physical Review Letters}, 105(15):150401, October 2010.

\bibitem[Gol65]{golden1965lower}
Sidney Golden.
\newblock Lower bounds for the helmholtz function.
\newblock {\em Physical Review}, 137(4B):B1127, 1965.

\bibitem[HHJ{\etalchar{+}}17]{haah2017sample}
Jeongwan Haah, Aram~W Harrow, Zhengfeng Ji, Xiaodi Wu, and Nengkun Yu.
\newblock Sample-optimal tomography of quantum states.
\newblock {\em IEEE Transactions on Information Theory}, 63(9):5628--5641,
  2017.

\bibitem[Hra97]{hradil_quantum-state_1997}
Z.~Hradil.
\newblock Quantum-state estimation.
\newblock {\em Physical Review A}, 55(3):R1561--R1564, March 1997.

\bibitem[HSSW97]{helmbold1997comparison}
David~P Helmbold, Robert~E Schapire, Yoram Singer, and Manfred~K Warmuth.
\newblock A comparison of new and old algorithms for a mixture estimation
  problem.
\newblock {\em Machine Learning}, 27(1):97--119, 1997.

\bibitem[LC17]{li_general_2017}
Yen-Huan Li and Volkan Cevher.
\newblock A {General} {Convergence} {Result} for the {Exponentiated} {Gradient}
  {Method}.
\newblock {\em arXiv:1705.09628 [math]}, May 2017.
\newblock arXiv: 1705.09628.

\bibitem[Low10]{low2010pseudo}
Richard~A Low.
\newblock Pseudo-randomness and learning in quantum computation.
\newblock {\em arXiv preprint arXiv:1006.5227}, 2010.

\bibitem[PC99]{pan1999complexity}
Victor~Y Pan and Zhao~Q Chen.
\newblock The complexity of the matrix eigenproblem.
\newblock In {\em Proceedings of the thirty-first annual ACM symposium on
  Theory of computing}, pages 507--516. ACM, 1999.

\bibitem[P{\v{R}}04]{QSTbook}
Matteo Paris and Jaroslav {\v{R}}eh{\'{a}}{\v{c}}ek, editors.
\newblock {\em Quantum State Estimation}.
\newblock Springer Berlin Heidelberg, 2004.

\bibitem[QHL{\etalchar{+}}13]{qi_quantum_2013}
Bo~Qi, Zhibo Hou, Li~Li, Daoyi Dong, Guoyong Xiang, and Guangcan Guo.
\newblock Quantum {State} {Tomography} via {Linear} {Regression} {Estimation}.
\newblock {\em Scientific Reports}, 3, December 2013.

\bibitem[{\v{R}}HKL07]{rehacek_diluted_2007}
Jaroslav {\v{R}}eh{\'a}{\v{c}}ek, Zden{\v{e}}k Hradil, E~Knill, and AI~Lvovsky.
\newblock Diluted maximum-likelihood algorithm for quantum tomography.
\newblock {\em Physical Review A}, 75(4):042108, April 2007.

\bibitem[SGS12]{smolin_efficient_2012}
John~A. Smolin, Jay~M. Gambetta, and Graeme Smith.
\newblock Efficient {Method} for {Computing} the {Maximum}-{Likelihood}
  {Quantum} {State} from {Measurements} with {Additive} {Gaussian} {Noise}.
\newblock {\em Physical Review Letters}, 108(7):070502, February 2012.

\bibitem[SZN17]{shang_superfast_2017}
Jiangwei Shang, Zhengyun Zhang, and Hui~Khoon Ng.
\newblock Superfast maximum-likelihood reconstruction for quantum tomography.
\newblock {\em Physical Review A}, 95(6):062336, June 2017.

\bibitem[Tho65]{thompson1965inequality}
Colin~J Thompson.
\newblock Inequality with applications in statistical mechanics.
\newblock {\em Journal of Mathematical Physics}, 6(11):1812--1813, 1965.

\bibitem[TRW05]{tsuda_matrix_2005}
Koji Tsuda, Gunnar R{\"a}tsch, and Manfred~K. Warmuth.
\newblock Matrix {Exponentiated} {Gradient} {Updates} for {On}-line {Learning}
  and {Bregman} {Projection}.
\newblock {\em Journal of Machine Learning Research}, 6(Jun):995--1018, 2005.

\bibitem[Ume62]{umegaki62}
H.~Umegaki.
\newblock {Conditional Expectation in an Operator Algebra}.
\newblock {\em Kodai Math. Sem. Rep.}, 14:59--85, 1962.

\bibitem[You18]{sourcecode}
Akram Youssry.
\newblock {MEG} online {QST}.
\newblock \url{https://github.com/akramyoussry/MEG\_online\_QST}, July 2018.

\end{thebibliography}
\end{document}